\documentclass[runningheads]{llncs}
\usepackage[T1]{fontenc}

\usepackage{graphicx}
\usepackage{lipsum}
\usepackage{parskip}
\usepackage{comment}
\usepackage{epsfig}
\usepackage{enumitem}
\usepackage{wrapfig}
\usepackage{hyperref}
\usepackage{amsmath,amssymb}

\usepackage[ruled,vlined,linesnumbered,titlenotnumbered,noend,resetcount,boxed]{algorithm2e}
\usepackage{lineno}
\usepackage{float}
\usepackage{caption}
\usepackage{tikz}
\usepackage{multirow}
\usepackage{listings}
\usepackage{circuitikz}
\usepackage{microtype}
\usepackage[capitalise]{cleveref}
    
\tikzset{>=stealth}
\ctikzset{tripoles/american and port/width=0.8, tripoles/american and port/height=.65}
\ctikzset{tripoles/american or port/width=0.8, tripoles/american or port/height=.65}
\usetikzlibrary{positioning,arrows,automata,shapes,decorations,decorations.markings,calc, matrix,decorations.pathmorphing, patterns,backgrounds}

\usetikzlibrary{arrows,automata,shapes,decorations,decorations.markings,calc, matrix,decorations.pathmorphing, patterns,backgrounds,shapes.misc,arrows.meta,positioning}

\usetikzlibrary{arrows.meta, positioning} 

\usepackage[normalem]{ulem} 
\usepackage{bm}
\usepackage{pgfplots}

\pgfdeclarelayer{bg}    
\pgfsetlayers{bg,main}  
\usetikzlibrary{fit,calc}

\usepackage{xspace}
\usepackage{booktabs}   
\usepackage{subcaption} 

\usepackage{makecell} 

\usepackage{pifont}
\usepackage{thm-restate}

\usepackage{algorithmic}

\usepackage{tabularx}
 \usepackage{multicol}

\usepackage{stfloats}

\newcolumntype{H}{>{\setbox0=\hbox\bgroup}c<{\egroup}@{}}

\usepackage{lipsum}

\input{mathsprograms.sty}
\usepackage{tikz}

\IncMargin{1em}

\newcounter{sarrow}

\newlist{myenum}{enumerate}{3} 
\setlist[myenum]{label=(\arabic*), nosep,leftmargin=*}
\crefname{myenumi}{Item}{Items}

\newlist{myitem}{itemize}{3} 
\setlist[myitem]{label=$\bullet$, nosep,leftmargin=*}
\crefname{myitem}{Item}{Items}

\colorlet{colorDEPS}{violet}
\colorlet{colorPO}{darkgray!80!black}
\colorlet{colorRF}{blue}
\colorlet{colorCO}{red!80!black}
\colorlet{colorFR}{purple}
\colorlet{colorECO}{orange}
\colorlet{colorCOM}{orange}
\colorlet{colorSW}{olive}
\colorlet{colorHB}{green!40!black}
\colorlet{colorPPO}{magenta}
\colorlet{colorXPPO}{magenta}
\colorlet{colorAPPO}{magenta}
\colorlet{colorRSEQ}{green!40!black}
\colorlet{colorSC}{violet}
\colorlet{colorPSC}{violet}
\colorlet{colorREL}{olive}
\colorlet{colorRMW}{brown}
\colorlet{colorPB}{olive}

\tikzset{
   every path/.style={>=stealth},
   po/.style={->,color=colorPO,thin,shorten >=-0.5mm,shorten <=-0.5mm},
   ppo/.style={->,color=colorPPO,thin,shorten >=-0.5mm,shorten <=-0.5mm},
   sw/.style={->,color=colorSW,shorten >=-0.5mm,shorten <=-0.5mm},
   rf/.style={->,color=colorRF,dashed,shorten >=-0.5mm,shorten <=-0.5mm},
   fr/.style={->,color=colorFR,dashed,shorten >=-0.5mm,shorten <=-0.5mm},
   hb/.style={->,color=colorHB,thick,shorten >=-0.5mm,shorten <=-0.5mm},
   co/.style={->,color=colorCO,dotted,very thick,shorten >=-0.5mm,shorten <=-0.5mm},
   rmw/.style={->,color=colorRMW,thick,shorten >=-0.5mm,shorten <=-0.5mm},
   rseq/.style={->,color=colorRSEQ,thick,dotted,shorten >=-0.5mm,shorten <=-0.5mm},
   sc/.style={->,color=colorSC,dashed,shorten >=-0.5mm,shorten <=-0.5mm},
   eco/.style={->,color=colorECO,thick,shorten >=-0.5mm,shorten <=-0.5mm},
}

\newcommand\inarr[1]{\begin{array}{@{}l@{}}#1\end{array}}
\newcommand\inarrC[1]{\begin{array}{@{}c@{}}#1\end{array}}

\newcommand\inparII[2]{\begin{array}{@{}l@{~~}||l@{~~}}\inarr{#1}&\inarr{#2}\end{array}}
\newcommand\inparIII[3]{\begin{array}{@{}l@{~~}||l@{~~}||l@{}}\inarr{#1}&\inarr{#2}&\inarr{#3}\end{array}}

\newcommand\mytag[1]{\tag{\textsf{#1}}}

\newcommand{\sett}[1]{\{#1\}}
\newcommand{\tup}[1]{\langle#1\rangle}

\newcommand{\R}{\mathsf{R}}
\newcommand{\W}{\mathsf{W}}
\newcommand{\E}{\mathsf{E}}

\newcommand{\F}{\mathsf{F}}

\newcommand{\T}{\mathsf{T}}
\newcommand{\tid}{\mathsf{tid}}

\newcommand{\loc}{\mathsf{loc}}

\newcommand{\ord}{\mathsf{ord}}
\newcommand{\sco}{\mathsf{sco}}
\newcommand{\cta}{\mathsf{cta}}
\newcommand{\gpu}{\mathsf{gpu}}
\newcommand{\sys}{\mathsf{sys}}
\newcommand{\wg}{\mathsf{cta}}

\newcommand{\lhb}{\mathsf{\color{colorHB}hb}}

\newcommand{\po}{\mathsf{\color{colorPO}po}}

\newcommand\rloc[1]{{{#1}_{=\loc}}}
\newcommand\rnloc[1]{{{#1}_{\neq\loc}}}

\newcommand{\poloc}{\rloc{\po}}
\newcommand{\npoloc}{\rnloc{\po}}
\newcommand{\hbloc}{\rloc{\lhb}}

\newcommand{\rf}{\mathsf{\color{colorRF}rf}}
\newcommand{\co}{\mathsf{\color{colorCO}co}}

\newcommand{\fr}{\mathsf{\color{colorFR}fr}}

\newcommand{\rmw}{\mathsf{\color{colorRMW}rmw}}

\newcommand{\rseq}{\mathsf{\color{blue}rseq}}

\newcommand{\sw}{\mathsf{\color{colorSW}sw}}

\newcommand{\prel}{\mathsf{\color{colorRSEQ}prel}}

\newcommand{\pacq}{\mathsf{\color{colorRSEQ}pacq}}

\newcommand{\fn}{\mathsf{f}}

\newcommand{\eco}{\mathsf{\color{colorECO}eco}}

\newcommand{\lloc}{\mathsf{loc}}
\newcommand{\Val}{\mathsf{Val}}

\newcommand{\Cif}{\mathbf{if}}
\newcommand{\Celse}{\mathbf{else}}

\newcommand{\MOna}{\textsc{na}}

\newcommand{\MOrlx}{\textsc{rlx}}
\newcommand{\MOacq}{\textsc{acq}}

\newcommand{\MOacqrel}{\textsc{acq-rel}}
\newcommand{\MOrel}{\textsc{rel}}
\newcommand{\MOsc}{\textsc{sc}}

\newcommand\SC {{\color{blue}\mathsf{sc}}}

\newcommand\scb {{\color{blue}\mathsf{scb}}}
\newcommand\psc {{\color{colorPSC}\mathsf{psc}}}
\newcommand\pscb {{\color{colorPSC}\mathsf{psc_{base}}}}
\newcommand\pscf {{\color{colorPSC}\mathsf{psc_{F}}}}

\newcommand{\irmw}{\mathsf{RMW}}

\newcommand{\cf}{\mathsf{cf}}

\newcommand{\eid}{\mathsf{id}}

\newcommand{\srcmm}{\mathsf{SRC11}}

\newcommand{\G}{\mathcal{G}}

\newcommand{\wt}{\mathsf{W}}
\newcommand{\rd}{\mathsf{R}}
\newcommand{\fence}{\mathsf{fnc}}
\newcommand{\barr}{\mathsf{bar}}

\definecolor{mygreen}{rgb}{0.05, 0.5, 0.06}

\newcommand{\hab}{\mathsf{\color{colorHB}hb}}
\newcommand{\incl}{\textcolor{teal}{\mathsf{incl}}}

\definecolor{darkorange}{rgb}{0.60, 0.35, 0.07}

\newcommand{\projname}{\mathsf{GPUMC}}

\newcommand{\blocked}{\mathsf{Blocked}}
\newcommand{\pgm}{\mathcal{P}}
\newcommand{\egraph}{\mathcal{G}}

\newcommand{\gorder}{\textcolor{brown!75!blue}{<_{{exe}}}}
\newcommand{\isoptimal}{\textsc{CheckOptimal}}
\newcommand{\addrf}{\texttt{addRF}}
\newcommand{\readfromlatest}{\texttt{readFromMOLatest}}

\newcommand{\reversed}{\texttt{reversed}}

\newcommand{\delayedrf}{\texttt{DelayedRf}}

\newcommand{\deletedevents}{\emph{Deleted}}

\newcommand{\pprfs}{\tt{ddRFs}}

\newcommand{\lastofg}[1]{\emph{last}(#1)}
\newcommand{\dartagnan}{\textsc{Dartagnan}}
\newcommand{\gpuverify}{\textsc{GPUVerify}}
\newcommand{\iguard}{\textsc{iGUARD}}
\newcommand{\gklee}{\textsc{G-Klee}}
\newcommand{\simulee}{\textsc{Simulee}}
\newcommand{\scordd}{\textsc{Scord}}

\newcommand{\prevstep}{\xleftarrow[]{P}}
\newcommand{\prevsteps}{\xleftarrow[]{P*}}
\newcommand{\nev}{\le_{nev}}
\newcommand{\nextevent}{\tt{NextEvent}}

\definecolor{myorange}{rgb}{0.93, 0.49, 0.1}
\definecolor{myred}{rgb}{0.82, 0.1, 0.26}
\definecolor{myblue}{rgb}{0.01, 0.28, 1.0}
\definecolor{myviolet}{rgb}{0.6, 0.4, 0.8}
\definecolor{mygray}{rgb}{0.9, 0.89, 0.89}
\definecolor{mypurple}{rgb}{0.41, 0.16, 0.38}

\newcommand\app\bullet

\newcommand\emptyword\epsilon

\newcommand\restrict[2]{#1\raise-.5ex\hbox{\ensuremath|}_{#2}}

\newcommand\conf\gamma

\newcommand\run\rho
\newcommand\thrun\pi
\newcommand\pth\pi

\newcommand\trace\tau
\newcommand\otrace\sigma

\newcommand\tlub\sqcup
\newcommand\tglb\sqcap
\newcommand\tequiv\sim
\newcommand\ttequiv\equiv
\newcommand\ctordering\sqSubset
\newcommand\tordering\sqsubseteq
\newcommand\stordering\sqsubset
\newcommand\eventset{{\tt E}}

\newcommand\inittranset{\tran_{\it init}}

\newcommand\varof[1]{{#1}.{\it var}}

\newcommand\id{{\it id}}

\newcommand\add\odot

\newcommand*\circledsmall[1]{\tikz[baseline=(char.base)]{
  \node[shape=circle,draw=none,fill=orange!20!white,inner sep=0.7pt, solid] (char) {\textcolor{black}{\texttt{#1}}};}}

\newcommand\lbl\ell
\newcommand\action\lbl
\newcommand\silentlbl\varepsilon

\newcommand\assigned\leftarrow

\newcommand\explore{\textsc{Explore}}
\newcommand\declarepostponed{\textsc{DelayedRFs}}
\newcommand\checkrace{\textsc{CheckAndRepairRace}}

\newcommand\schedule\beta

\newcommand\setname[1]{{\mathcal A}}

\newcommand\obs\alpha
\newcommand\obsseq\pi
\newcommand\obsseqsub\preceq
\newcommand\obsseqminus\ominus

\crefformat{section}{\S#2#1#3} 
\crefformat{subsection}{\S#2#1#3}
\crefformat{subsubsection}{\S#2#1#3}
\crefformat{subsubsection}{\S#2#1#3}
\crefformat{appendix}{Appendix~#2#1#3}
\crefformat{subappendix}{Appendix~#2#1#3}
\crefformat{subsubappendix}{Appendix~#2#1#3}

\definecolor{mGreen}{rgb}{0,0.6,0}
\definecolor{mGray}{rgb}{0.5,0.5,0.5}
\definecolor{mPurple}{rgb}{0.58,0,0.82}
\definecolor{backgroundColour}{rgb}{0.95,0.95,0.92}

\lstdefinestyle{CStyle}{
    backgroundcolor=\color{white},   
    commentstyle=\color{mGreen},
    keywordstyle=\color{magenta},
    numberstyle=\tiny\color{mGray},
    stringstyle=\color{mPurple},
    basicstyle=\linespread{0.6}\footnotesize,
    breakatwhitespace=false,         
    breaklines=true,                 
    captionpos=b,                    
    keepspaces=true,                 
    numbers=left,                    
    numbersep=5pt,                  
    showspaces=false,                
    showstringspaces=false,
    showtabs=false,                  
    tabsize=2,
    language=C
}

\SetAlgoNoLine
\SetInd{0.8em}{0.8em}
\SetNoFillComment

\SetCommentSty{mycommfont}
\DontPrintSemicolon

\definecolor{lime}{HTML}{A6CE39}
\DeclareRobustCommand{\orcidicon}{%
	\begin{tikzpicture}
	\draw[lime, fill=lime] (0,0) 
	circle [radius=0.16] 
	node[white] {{\fontfamily{qag}\selectfont \tiny ID}};
	\draw[white, fill=white] (-0.0625,0.095) 
	circle [radius=0.007];
	\end{tikzpicture}
	\hspace{-2mm}
}

\foreach \x in {A, ..., Z}{%
	\expandafter\xdef\csname orcid\x\endcsname{\noexpand\href{https://orcid.org/\csname orcidauthor\x\endcsname}{\noexpand\orcidicon}}
}

\begin{document}

\title{$\projname$: A Stateless Model Checker for GPU Weak Memory Concurrency
}
\author{
Soham Chakraborty\inst{1,2} \orcidA{}  \and
S. Krishna  \inst{2} \orcidB{} \and
Andreas Pavlogiannis \inst{3} \orcidC{} \and 
Omkar Tuppe \inst{2} \orcidD{} }
\authorrunning{S. Chakraborty et al.}
%
\institute{TU Delft, Netherlands \\
\email{s.s.chakraborty@tudelft.nl}
\and
IIT Bombay, India\\
\email{\{krishnas,omkarvtuppe\}@cse.iitb.ac.in}
\and
Aarhus University, Denmark\\
\email{pavlogiannis@cs.au.dk}
}

\maketitle


\begin{abstract}
GPU computing is embracing weak memory concurrency for performance improvement. However, compared to CPUs, modern GPUs provide more fine-grained concurrency features such as scopes, have additional properties like divergence, and thereby follow different weak memory consistency models. These features and properties make concurrent programming on GPUs more complex and error-prone. 
To this end, we present $\projname$, a stateless model checker to check the correctness of GPU shared-memory concurrent programs under scoped-RC11 weak memory concurrency model.
$\projname$ explores all possible executions in GPU programs to reveal various errors - races, barrier divergence, and assertion violations. In addition, $\projname$ also automatically repairs these errors in the appropriate cases. 

We evaluate $\projname$ on benchmarks and real-life GPU programs. $\projname$ is efficient both in time and memory in verifying large GPU programs where state-of-the-art tools are timed out. In addition,  $\projname$ identifies all known errors in these benchmarks compared to the state-of-the-art tools.

\end{abstract}

\section{Introduction}\label{sec:introduction}

In recent years GPUs have emerged as mainstream processing units, more than just accelerators~\cite{8916327_survey_ml,pandey2022transformational_drugdiscovery,ozerk2022efficient_encrypt,francis2014autonomous}. 
Modern GPUs provide support for more fine-grained shared memory access patterns, allowing programmers to optimize performance beyond the traditional lock-step execution model typically associated with SIMT architectures.
To this end, GPU programming languages such as CUDA and OpenCL~\cite{NVIDIA_CUDA_v12.6,opencl}, as well as libraries~\cite{cutlass,cccl}, have adopted C/C++ shared memory concurrency primitives.

Writing correct and highly efficient shared-memory concurrent programs is already a challenging problem, even for CPUs. 
GPU concurrency poses further challenges. 
Unlike CPU threads, the threads in a GPU are organized hierarchically and synchronize via barriers during execution. 
Moreover, shared-memory accesses are \emph{scoped}, resulting in more fine-grained rules for synchronization, based on the proximity of their threads. 
Although these primitives and rules play a key role in achieving better performance, they are also complex and prone to errors. 

 GPU concurrency may result in various types of concurrency bugs -- assertion violations, data races, heterogeneous races, and barrier divergence. 
 While assertion violations and data race errors are well-known in CPU concurrency, they manifest in more complicated ways in the context of GPU programs. 
 The other two types of errors, heterogeneous races and barrier divergence, are GPU specific. 
 To catch these errors, it is imperative to explore all possible executions of a program. 
 
The set of possible executions of a GPU concurrent program is determined by its underlying consistency model. 
State-of-the-art architectures including GPUs follow weak consistency, and as a result a program may exhibit extra behaviors in addition to the interleaving executions or more formally sequential consistency (SC) \cite{lamport-sc}. 
However, as the weak memory concurrency models in GPUs differ from the ones in the CPUs, the state-of-the-art analysis and verification approaches for programs written for CPUs do not suffice in identifying these errors under GPU weak memory concurrency. 
As a result, automated reasoning of GPU concurrency, particularly under weak consistency models, even though a timely and important problem, has remained largely unexplored.

To address this gap, in this paper we develop the $\projname$ model checker for a scoped-C/C++ programming languages~\cite{Lustig:2019} for GPUs. 
Scoped-C/C++ has all the shared memory access primitives provided by PTX and Vulkan, and in addition, provide SC memory accesses. 
The recent work of \cite{Lustig:2019} formalizes the scoped C/C++ concurrency in scoped-RC11 memory model ($\srcmm$), similarly to the formalization of C/C++ concurrency in RC11~\cite{Lahav:2017}. 
Consequently, $\projname$ is developed for the $\srcmm$ model. 
The consistency properties defined by $\srcmm$, 
scoped C/C++ programming language follows catch fire semantics similar to traditional C/C++, that is, a program having a $\srcmm$ consistent execution with a data race has undefined behavior. 
In addition, scoped C/C++ defines \emph{heterogeneous race} \cite{Lustig:2019,hernan24,gaster2015hrf,hower2014heterogeneous} based on the scopes of the accesses, and a program having a $\srcmm$-consistent execution with heterogeneous race also has undefined behavior.

Stateless Model Checking (SMC) is a prominent automated verification technique \cite{DBLP:conf/popl/ClarkeES83} that explores all possible executions of a program in a systematic manner. However, the number of executions can grow exponentially larger in the number of concurrent threads, which poses a key challenge to a model checker. 
To address this challenge, partial order reduction (POR)  \cite{DBLP:journals/sttt/ClarkeGMP99,DBLP:books/sp/Godefroid96,DBLP:conf/cav/Peled93} and subsequently dynamic partial order reduction (DPOR) techniques have been proposed \cite{DBLP:conf/popl/FlanaganG05}. 
More recently, several DPOR algorithms are proposed for different weak memory consistency models to explore executions in a time and space-efficient manner \cite{DBLP:journals/pacmpl/Kokologiannakis18,10.1145/2806886,10.1145/3276505,DBLP:journals/jacm/AbdullaAJS17,10.1007/978-3-662-46681-0_28,10.1145/2737924.2737956}. 
For instance, GenMC-Trust \cite{DBLP:journals/pacmpl/Kokologiannakis22} and POP \cite{abdulla2024parsimonious} are recently proposed polynomial-space DPOR algorithms.
While these techniques are widely applied for programs written for CPUs (weak memory) concurrency models \cite{DBLP:journals/pacmpl/Kokologiannakis18,10.1145/2806886,10.1145/3276505,DBLP:journals/jacm/AbdullaAJS17,DBLP:conf/pldi/Kokologiannakis19,DBLP:journals/pacmpl/Kokologiannakis22}, to our knowledge, DPOR-based model checking has not been explored for GPU weak memory concurrency.

$\projname$ extends the GenMC-TruSt~\cite{DBLP:journals/pacmpl/Kokologiannakis22} approach to handle the GPU-specific features that the original GenMC lacks. 
More specifically, $\projname$ implements an exploration-optimal, sound, and complete DPOR algorithm with linear memory requirements that is also parallelizable. 
Besides efficient exploration, $\projname$ detects all the errors discussed above and automatically repairs certain errors such as heterogenous races. 
Thus $\projname$ progressively transforms a heterogeneous-racy program to generate a heterogeneous-race-free version. 
We empirically evaluate $\projname$ on several benchmarks to demonstrate its effectiveness. The benchmarks range from small litmus tests to real applications, used in GPU testing \cite{Sorensen:2021,gpuharbor}, bounded model checking \cite{dat3m}, and verification under sequential consistency \cite{Kamath:2020,iguard}. 
$\projname$ explores the executions of these benchmarks in a scalable manner and identifies the errors. 
We compare $\projname$ with $\dartagnan$ \cite{hernan24}, a bounded model checker for GPU weak memory concurrency \cite{dat3m}. 
$\projname$ identifies races which are missed by $\dartagnan$ in its benchmarks and also outperforms $\dartagnan$ significantly in terms of memory and time requirements in identifying concurrency errors. 

\textbf{Contributions \& outline} To summarize, the paper makes the following contributions. \Cref{sec:overview} and \Cref{sec:semantics} provide an overview of GPU weak memory concurrency and its formal semantics. Next, \Cref{sec:dpor} and \Cref{sec:evaluation} discuss the proposed DPOR algorithm and its experimental evaluation. Finally, we discuss the related work in \Cref{sec:related} and conclude in \Cref{sec:conclusion}.

\section{Overview of GPU Concurrency}
\label{sec:overview}

A shared memory GPU program consists of a fixed set of threads with a set of shared memory locations and thread-local variables. 
Unlike in the CPU, the GPU threads are structured in hierarchies at multiple levels: cooperative thread array(CTA) ($\wg$), GPU ($\gpu$), and system ($\sys$),       
where $\wg$ is a collection of threads and $\gpu$ is a group of $\cta$, and finally $\sys$ consists of a set of $\gpu$s and threads of other devices such as CPUs. Thus, a thread can be identified by its ($\wg$, $\gpu$) identifiers and its thread identifier. 
The system (sys) is the same for all threads. 

Shared memory operations are one of read, write, atomic read-modify-write (RMW), fence ($\fence$) or barrier ($\barr$).  
Similar to the C/C++ concurrency \cite{cstandard,cppstandard}, these accesses are non-atomic read or write, or 
atomic accesses with memory orders. Thus accesses are classified as: non-atomic ($\MOna$), relaxed ($\MOrlx$), acquire ($\MOacq$), release ($\MOrel$), acquire-release ($\MOacqrel$), or sequentially consistent ($\MOsc$). 
In increasing strength,  $\MOna \sqsubset \MOrlx \sqsubset \set{\MOrel,\MOacq} \sqsubset \MOacqrel \sqsubset \MOsc$. 

The shared memory accesses of the GPU are further parameterized with a scope $\sco \in \set{\wg,\gpu,\sys}$. The scope of an operation determines its role in synchronizing with other operations in other threads based on proximity. Thus, shared memory accesses are of the following form where $o_r$, $o_w$, $o_u$, $o_f$ denote the memory orders of the read, write, RMW, and fence accesses respectively.
\[
\inarr{
r\!=\!X^\sco_{o_r} \!\mid\! X^\sco_{o_w}\!=\! E \!\mid\! r \!=\! \irmw^\sco_{o_u}(X, \!E_r,\!E_w) \!\mid\!\! \fence^\sco_{o_f} \!\mid\! \barr^\sco(\eid)     
}
\]
A read access $r\!=\!X^\sco_{o_r}$ returns the value of shared memory location/variable  $X$ to thread-local variable $r$ with memory order $o_r$ selected from $\sett{\MOna, \MOrlx, \MOacq, \MOsc}$. 
A write access  $X^\sco_{o_w}\!=\! E$ writes the value of expression $E$ to the location $X$ with memory order $o_w$ selected from $\sett{\MOna, \MOrlx, \MOrel, \MOsc}$. 
The superscript $\sco$ refers to the scope. An RMW access  $r \!=\! \irmw^\sco_{o_u}(X,E_r,E_w)$, atomically updates the value of location $X$ with the value of $E_w$ 
if the read value of $X$ is $E_r$. On failure, it performs only the read operation. The memory order of an RMW is $o_u$ selected from $\sett{\MOrlx, \MOrel, \MOacq, \MOacqrel, \MOsc}$. A fence access $\fence$ is performed with a memory order $o_f$ selected from $\sett{\MOrel, \MOacq, \MOacqrel, \MOsc}$. GPUs also provide barrier operations where a set of threads synchronize and therefore affect the behaviors of a program.
For a barrier operation $\barr^\sco(\eid)$, $\textrm{sco}$ refers to the scope of the barrier and $id$ denotes the barrier identifier.
We model barriers as acquire-release RMWs ($\irmw^\sco_\MOacqrel$) parameterized with scope $\sco$ on a special auxiliary variable (similar to \cite{kokologiannakis2021bam}). 

\begin{figure}[t]
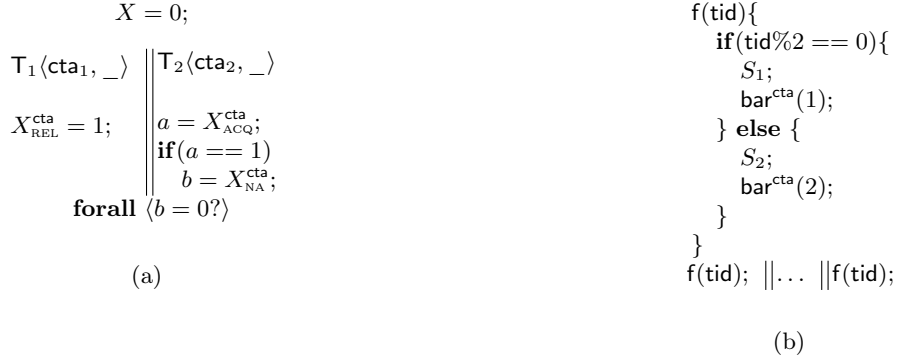

{
\centering
\begin{subfigure}[t]{0.30\textwidth}
\[
\inarrC{
X=0;
\\[2ex]
\inparII{
\T_1\tup{\wg_1, \_}
\\[3ex]
X^\wg_\MOrel = 1; 
\\
\\[2ex]
}{
\T_2\tup{\wg_2, \_}
\\[3ex]
a=X^\wg_\MOacq;
\\ \Cif(a==1)
\\ \quad b = X^\wg_\MOna; 
}
\\[3ex]
{\bf{forall}}~\tup{b=0?}
}
\]    
\caption{}
\label{fig:assertrace}
\end{subfigure}
\hfill 
\begin{subfigure}[t]{0.30\textwidth}
\[
\inarrC{
\inarr{
\fn(\tid) \{
\\ \quad \Cif(\tid \% 2 == 0) \{
\\     \qquad S_1; 
\\  \qquad \barr^\wg(1);
\\  \quad \} \ \Celse \ \{
\\  \qquad    S_2; 
\\  \qquad  \barr^\wg(2);
\\ \quad \}
\\ \}
}
\\[5ex]
\inparIII{
\fn(\tid);
}{
\ldots
}{
\fn(\tid);
}
}
\]
\caption{}
\label{fig:divergence}
\end{subfigure}
\caption{Example of GPU concurrency errors. In (a), we have two threads $T_1, T_2$ from the CTAs $\wg_1, \wg_2$. 
In (b) all threads are in the same CTA.}
\label{fig:GPU:errors}
}
\end{figure}


\subsection{GPU Concurrency Errors}
\label{sec:GPU:errors}
Traditionally,  two key errors in shared memory concurrency are assertion violations and data races.   In addition, concurrent programs for GPUs may contain heterogeneous races and barrier divergence errors. The behavior of a program with data race or  heterogeneous race is undefined, while  divergence errors may lead to deadlocks \cite[Section 16.6.2]{NVIDIA_CUDA_v12.6}, \cite{Lustig:2019}, \cite{hernan24}. 

\textit{Assertion violation}: In our benchmarks assertion violations imply weak memory bugs. Assertions verify the values of the variables and memory locations in a program. If the intended values do not match, it results in an assertion violation. 
    Consider the program in \Cref{fig:assertrace} having the assertion ${\bf{forall}}~b=0?$ which checks whether, for all executions,  $b$ is 0.  
    If the value of $X$ read into $a$ in $\T_2$ is 1, then $b$ cannot read a stale value 0 from $X$ and the assertion fails.  
    
    \textit{Data race}:  Two operations $a$ and $b$ in an execution are said to be in a data race \cite{Lustig:2019} \cite{hernan24} if (i) $a$ and $b$ are concurrent, that is, not related by \emph{happens-before}, (ii) they access the same memory location, (iii) at least one of the accesses is a write operation, and (iv) at least one of the accesses is a non-atomic operation. 
    In \Cref{fig:assertrace}, if $\wg_1=\wg_2$,   the threads are in the same $\cta$. In that case, if the acquire-read of $X$ in the second thread reads from the release-write in the first thread, then it establishes synchronization. Hence, the release-write of $X$ \emph{happens-before} the non-atomic read of $X$, and the program has no data race. 
    
    \textit{Heterogeneous race}: 
    Two operations $a$ and $b$ in an execution are in a heterogeneous race if (i) $a$ and $b$ are concurrent, (ii) they access the same memory location, (iii) at least one of the accesses is a write operation, and (iv) both accesses are atomic with non-inclusive scope, that is, the scopes of each access includes the thread executing the other
access.
Note that a heterogeneous race may take place between atomic accesses. 
    In \Cref{fig:assertrace}, if $\wg_1 \neq \wg_2$ then the acquire-read and release-write do not synchronize and 
    consequently are in a heterogeneous race. Then the program also has a data race between the non-atomic read of $X$ and release-write of $X$. 
    
    \textit{Barrier divergence}:  Given a barrier, the threads within the given scope of the barrier synchronize. During execution, while a thread reaches the barrier, it waits for all the other threads to reach the barrier before progressing the execution further.  
    Consider the program in \Cref{fig:divergence}, where all threads execute the function $\fn()$. The threads with even thread identifiers synchronize to $\barr(1)$ and the thread with odd thread identifiers synchronize to $\barr(2)$. 
    Hence the threads are diverging and not synchronizing to a single barrier. Modern GPUs consider it as a divergence error as the non-synchronizing threads may result in a deadlock. 
         Following the definition from \cite[Section 16.6.2]{NVIDIA_CUDA_v12.6}, we report barrier divergence if at least one of the threads participating in the barrier is blocked at the barrier at the end of execution (no next instruction to execute).


\section{Formal Semantics}
\label{sec:semantics}
In this section, we elaborate on the formal semantics of GPU concurrency. 
A program's semantics is formally represented by a set of \emph{consistent} executions. An execution consists of a set of events and various relations between the events. 

\noindent{\bf Events} 
An event corresponds to the effect of executing a shared memory or fence access in the program. 
An event $e =\langle$$\id,$$ \tid,$$ \op,$$ \lloc,$$ \ord,$$ \sco,$$ \Val$$\rangle$ is 
represented by a tuple where $\id$, $\tid$, $\op$, $\lloc$, $\ord$, $\sco$, $\Val$ denote the
event identifier, thread identifier, memory operation, memory location accessed, memory order, scope, read or written value. 
A read, write, or fence access generates a read, write, or fence event. A successful RMW generates a pair of read and write events and a failed RMW generates a read event.
A read event $\rd^\sco_o(X,v)$ reads from location $X$ and returns value $v$ with memory order $o$ and scope $\sco$. A write event $\wt^\sco_o(X,v)$ writes value $v$ to location $X$ with memory order $o$ and scope $\sco$. 
A fence event $\F^\sco_o$ has memory order $o$ and scope $\sco$. Note that for a fence event, $\lloc=\Val=\bot$. 
The set of read, write, and fence events are denoted by $\rd$, $\wt$, and $\F$ respectively. 

\noindent {\bf Relations}
The events of an execution are associated with various relations.
The relation program-order ($\po$) denotes the syntactic order among the events. In each thread $\po$ is a total order. The relation reads-from ($\rf$) relates a pair of same-location write and read events $w$ and $r$ having the same values to denote that $r$ has read from $w$. Each read has a unique write to read from ($\rf^{-1}$ is a function). 
The relation coherence order ($\co$) is a total order on the same-location write events. 
The relation $\rmw$ denotes a successful RMW operation that relates a pair of same-location read and write events $r$ and $w$ which are in \emph{immediate}-$\po$ relation, that is, no other event $a$ exists such that $(r,a)$ and $(a,w)$ are in $\po$ relations. 
We derive new relations following the notations below.

\noindent\textbf{Notation on relations} Given a binary relation $B$, we write $B^{-1}$, $B^?$, $B^+$, $B^*$ to denote its inverse, reflexive, transitive, reflexive-transitive closures respectively. We compose two relations $B_1$ and $B_2$ by $B_1;B_2$. Given a set $A$, [A] denotes the identity relation on the set $A$. 
Given a relation $B$, we write $\rloc{B}$ and $\rnloc{B}$ to denote relation $B$ on same-location and different-location events respectively. For example, $\poloc$ relates a pair of same-location events that are $\po$-related. Similarly, $\npoloc$ relates $\po$-related events that access different locations. Relation from-read ($\fr$) relates a pair of same-location read and write events $r$ and $w'$.
If $r$ reads from $w$ and $w'$ is $\co$-after $w$ then $r$ and $w'$ are in $\fr$ relation: $\fr \triangleq \rf^{-1};\co$.

\noindent{\bf Execution \& consistency}
An execution is a tuple $\egraph=\tup{\E,\po,\rf,\co,\rmw}$ consisting of a set of events $\E$, and the sets of $\po$, $\rf$, $\co$, and $\rmw$ relations. We represent an execution as a graph where the nodes represent events and different types of edges represent respective relations. 
A concurrency model defines a set of axioms or constraints based on the events and relations. If an execution satisfies all the axioms of a memory model then the execution is consistent in that memory model.

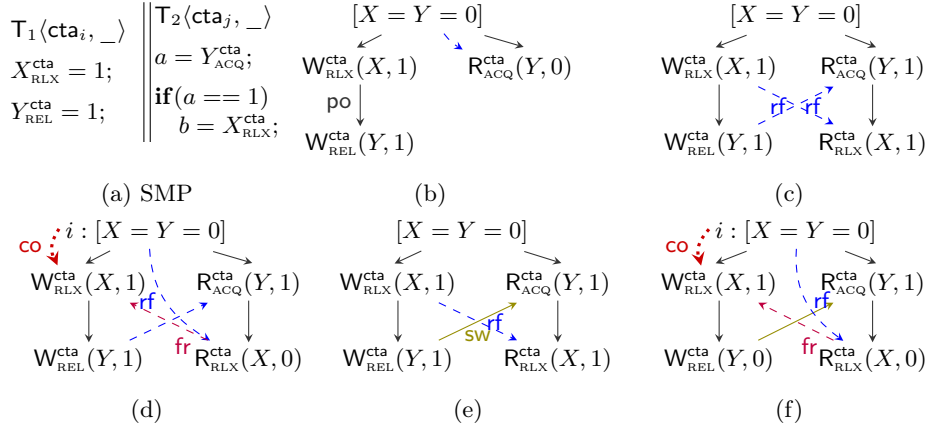
\begin{figure*}[t]
\begin{subfigure}{0.3\textwidth}
    \centering
    {
\[
\inarrC{
\inparII{
\T_1\tup{\wg_i, \_}
\\[1.2ex]
X^\wg_\MOrlx = 1;
\\[1.1ex]
Y^\wg_\MOrel = 1;
}{
\T_2\tup{\wg_j, \_}
\\[1.2ex]
a=Y^\wg_\MOacq;
\\[1.1ex]
\Cif(a==1) 
\\
\quad b=X^\wg_\MOrlx;
}
}
\]
}
\caption{SMP\label{prg:smp}}
\end{subfigure}
    \centering
\begin{subfigure}{0.3\textwidth}
 \centering
     \begin{tikzpicture}[yscale=0.68,xscale=0.85]
     {
      \node (s) at (0,0)  {$[X=Y=0]$ };
      \node (s11) at (-0.9,-1)  {$\W^\wg_\MOrlx(X,1)$};
      \node (s12) at (-0.9,-2.5)  {$\W^\wg_\MOrel(Y,1)$};
      \node (s21) at (1.6,-1)  {$\R^\wg_\MOacq(Y,0)$};

      \draw[po] (s) edge (s11.north) edge (s21.north);
      \draw[po] (s11) edge node[left]{$\po$} (s12);
      \draw[rf,bend right=10] (s) edge (s21); }
    \end{tikzpicture}   
     \caption{}
 \label{fig:ex1:smp}    
\end{subfigure}
\hfill
\begin{subfigure}{0.3\textwidth}
 \centering
     \begin{tikzpicture}[yscale=0.68,xscale=0.85] 
     {
      \node (s) at (0,0)  {$[X=Y=0]$ };
      \node (s11) at (-1.2,-1)  {$\W^\wg_\MOrlx(X,1)$};
      \node (s12) at (-1.2,-2.5)  {$\W^\wg_\MOrel(Y,1)$};
      \node (s21) at (1.2,-1)  {$\R^\wg_\MOacq(Y,1)$};
      \node (s22) at (1.2,-2.5)  {$\R^\wg_\MOrlx(X,1)$};

      \draw[po] (s) edge (s11.north) edge (s21.north);
      \draw[po] (s11) edge (s12);
      \draw[po] (s21) edge (s22);
      \draw[rf] (s12) edge  node[left]{$\rf$} (s21);
       \draw[rf,bend right=0] (s11) edge node[right]{$\rf$} (s22);  }
    \end{tikzpicture}  
     \caption{}
 \label{fig:ex2:smp}       
\end{subfigure}
\hfill
\begin{subfigure}{0.3\textwidth}
 \centering
     \begin{tikzpicture}[yscale=0.68,xscale=0.85]
     {
      \node (s) at (0,0)  {$i:[X=Y=0]$ };
      \node (s11) at (-0.9,-1)  {$\W^\wg_\MOrlx(X,1)$};
      \node (s12) at (-0.9,-2.5)  {$\W^\wg_\MOrel(Y,1)$};
      \node (s21) at (1.6,-1)  {$\R^\wg_\MOacq(Y,1)$};
      \node (s22) at (1.6,-2.5)  {$\R^\wg_\MOrlx(X,0)$};

      \draw[po] (s) edge (s11.north) edge (s21.north);
      \draw[po] (s11) edge (s12);
      \draw[po] (s21) edge (s22);
      \draw[rf] (s12) edge (s21);
      \draw[rf,bend right=25] (s) edge node[left]{$\rf$}(s22);
      \draw[fr] (s22) edge node[below,pos=0.3]{$\fr$} (s11);
      \draw[co,bend right=25] (s.west) to node[left]{$\co$} (s11); }
    \end{tikzpicture}   
      \caption{}
  \label{fig:ex3:smp}       
\end{subfigure}
\hfill 
\begin{subfigure}{0.3\textwidth}
 \centering
     \begin{tikzpicture}[yscale=0.68,xscale=0.85]
     {
      \node (s) at (0,0)  {$[X=Y=0]$ };
      \node (s11) at (-1,-1)  {$\W^\wg_\MOrlx(X,1)$};
      \node (s12) at (-1,-2.5)  {$\W^\wg_\MOrel(Y,1)$};
      \node (s21) at (1.5,-1)  {$\R^\wg_\MOacq(Y,1)$};
      \node (s22) at (1.5,-2.5)  {$\R^\wg_\MOrlx(X,1)$};

      \draw[po] (s) edge (s11.north) edge (s21.north);
      \draw[po] (s11) edge (s12);
      \draw[po] (s21) edge (s22);
      \draw[sw] (s12) edge node[below]{$\sw$} (s21);
       \draw[rf,bend right=0] (s11) edge node[right]{$\rf$}  (s22); }
    \end{tikzpicture}  
     \caption{}
 \label{fig:ex4:smp}       
\end{subfigure}
\hfill 
\begin{subfigure}{0.3\textwidth}
 \centering
     \begin{tikzpicture}[yscale=0.68,xscale=0.85]
     {
      \node (s) at (0,0)  {$i:[X=Y=0]$ };
      \node (s11) at (-1.2,-1)  {$\W^\wg_\MOrlx(X,1)$};
      \node (s12) at (-1.2,-2.5)  {$\W^\wg_\MOrel(Y,0)$};
      \node (s21) at (1.2,-1)  {$\R^\wg_\MOacq(Y,1)$};
      \node (s22) at (1.2,-2.5)  {$\R^\wg_\MOrlx(X,0)$};

      \draw[po] (s) edge (s11.north) edge (s21.north);
      \draw[po] (s11) edge (s12);
      \draw[po] (s21) edge (s22);
      \draw[sw] (s12) edge (s21);
      \draw[rf,bend right=25] (s) edge node[right]{$\rf$} (s22);
      \draw[fr] (s22) edge node[below,pos=0.3]{$\fr$} (s11);
      \draw[co,bend right=25] (s.west) to node[left]{$\co$} (s11); }
    \end{tikzpicture}   
      \caption{}
  \label{fig:ex5:smp}       
\end{subfigure}

\caption{
Executions shown in (b) and (c) are independent of whether $i=j$ or not.
(b) shows an execution where $Y$ reads 0 from the initial location.
(c) shows an execution where $Y$ and $X$ read 1 in $\T_2$.
(d) shows an execution where $Y$ reads 1 from $\T_1$ but cannot synchronize, as $\T_1$ and $\T_2$ are
in different CTAs ($i\neq j$). 
If $i\neq j$, $X$ may read 0 from initialization.
(e) is a special case of execution shown in (c) where $i = j$. If $ i==j$, then read and write on $Y$ are in
synchronization relation because these accesses on $Y$ are scope-inclusive. 
(f) shows an execution where there is a synchronization on $Y$ with an inclusion
relation (so again $i=j$). 
Hence, $X$ in $\T_2$ cannot read value 0 from initialization, as it violates the coherence axiom; consequently, the execution is forbidden.
}
\label{fig:smp:srcmm}
\end{figure*}


\paragraph{\bf $\srcmm$ consistency model}
\label{sec:rcmm}


We first explain the relations of the RC11 model \cite{Lahav:2017} which is extended to $\srcmm$ \cite{Lustig:2019} for GPUs, defined in \Cref{fig:scr11}.

\begin{figure}[t]
\centering
\begin{minipage}{0.56\textwidth}
\centering
{
\[
\inarr{
\rseq \triangleq [\wt];\poloc^?;[\wt_{\sqsupseteq \MOrlx}];((\incl \cap \rf);\rmw)^*
\\[1.02ex]
\prel \triangleq [\E_{\sqsupseteq\MOrel}];([\F];\po)^?
\\[1.02ex] 
\pacq \triangleq (\po;[\F])^?;[\E_{\sqsupseteq \MOacq}]
\\[1.02ex]
\sw \triangleq \prel;\rseq;(\incl \cap \rf);\pacq
\\[1.02ex]
\lhb \triangleq (\po \cup (\incl \cap \sw))^+
\\[1.02ex]
\scb \triangleq \po \cup (\npoloc;\lhb;\npoloc) \cup \hbloc \cup \co \cup \fr
\\[1.02ex] 
\pscb \triangleq ([\E_\MOsc] \cup [\F_\MOsc];\lhb^?);\scb;([\E_\MOsc] \cup \lhb^?;[\F_\MOsc])
\\[1.02ex]
\pscf \triangleq [\F_\MOsc];(\lhb \cup \lhb;\eco;\lhb);[\F_\MOsc]
\\[1.02ex]
\psc \triangleq \pscb \cup \pscf
}
\]
}
\end{minipage}%
\hfill\vline\hfill
\begin{minipage}{0.42\textwidth}
\centering
\begin{subfigure}{\textwidth}
\centering
\begin{minipage}{0.45\textwidth}   
\flushleft
\begin{tikzpicture}[yscale=0.6]
{
      \node (a) at (0,0)  {$a$};
      \node (b) at (0,-1.5)  {$b$};
      \draw[hb,bend right=30] (a) edge node[left]{$\lhb$} (b);
      \draw[eco,bend right=30] (b) edge node[right,pos=0.2]{$\eco^?$} (a); }
\end{tikzpicture}
\end{minipage}
\hfill
\begin{minipage}{0.45\textwidth}    
\flushright
\begin{tikzpicture}[yscale=0.6]
{
      \node (a) at (0,0)  {$r$};
      \node (b) at (0,-1.5)  {$w$};
      \node (c) at (1,-0.7)  {$w'$};
      \draw[rmw] (a) edge node[left,pos=0.2]{$\rmw$} (b);
      \draw[fr] (a) edge node[above]{$\fr$} (c);
      \draw[fr] (c) edge node[below]{$\co$} (b); }
\end{tikzpicture}
\end{minipage}
\label{fig:scr11 cycles}    
\end{subfigure}
\vfill
\begin{subfigure}{\textwidth}
\centering
\begin{minipage}{0.4\textwidth}   
\flushright
\begin{tikzpicture}[yscale=0.6]
      \node (a) at (0,0)  {$a$};
      \node (b) at (1.5,0)  {$b$};
      \draw[->,colorPSC,bend right=30](a) edge node[below]{$\incl \cap \psc$} (b);
      \draw[co,bend right=30](b) edge node[above]{$\co$}(a);
      
\end{tikzpicture}
\end{minipage}
\hfill
\begin{minipage}{0.45\textwidth}    
\flushright
\begin{tikzpicture}[yscale=0.6]
    {
      \node (a) at (0,0)  {$r$};
      \node (b) at (0,-1.5)  {$w$};
      \node (c) at (1.5,0)  {$r'$};
      \node (d) at (1.5,-1.5)  {$w'$};
      \draw[po] (a) edge node[left]{$\po$} (b);
      \draw[rf] (b) edge node[below]{$\rf$} (c);
      \draw[po] (c) edge node[left=-1pt]{$\po$} (d);
      \draw[rf] (d) edge  (a); }
\end{tikzpicture}
\end{minipage}    
\end{subfigure}
\vfill
\begin{subfigure}{\textwidth}
\vline
\begin{itemize}[leftmargin=*]
\item[-] $\lhb;\eco^?$ is irreflexive \hfill {\small(Coherence)}  
\smallskip
\item[-] $\rmw\cap (\fr;\co)$ is empty \hfill {\small(Atomicity)}
\item[-] $(\incl \cap \psc)$ is acyclic \hfill {\small(SC)}
\smallskip
\item[-] $(\po \cup \rf)$ is acyclic \hfill {\small(No-Thin-Air)}
\end{itemize}
\end{subfigure}
\end{minipage}
 \caption{$\srcmm$ relations and axioms with some violation patterns. \label{fig:scr11}}
\end{figure}
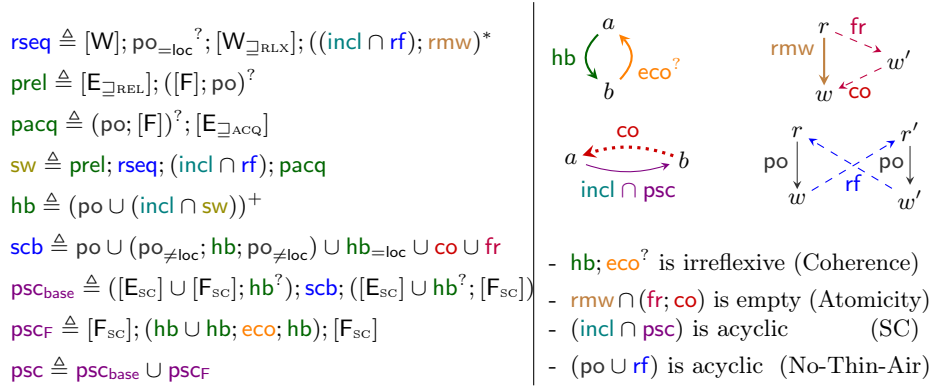

\noindent\textbf{RC11 relations}
Relation extended-coherence-order ($\eco$) is a transitive closure of the read-from ($\rf$), coherence order ($\co$), and from read ($\fr$) relations, that is, $\eco \triangleq (\rf \cup \co \cup \fr)^+$. Note that the $\eco$ related events always access the same memory location. 

Relation synchronizes-with ($\sw$) relates a release event to an acquire event. 
For example, when an acquire read reads from a release write then the pair establishes an $\sw$ relation. 
In general, $\sw$ uses release-sequence $\rseq$ that starts at a release store or fence event and ends at an acquire load or fence event with an intermediate chain of $\rf$-related $\rmw$ relations. 
Finally, relation happens-before ($\lhb$) is the transitive closure of the $\po$ and $\sw$ relations. 

To relate the SC memory accesses and fences, the RC11 model defines the $\scb$ relation. A pair of events $a$ and $b$ is in $\scb$ relation in one of these cases: (1) $(a,b)$ is in $\po$, $\co$, or $\fr$ relation. (2) $a$ and $b$ access the same memory location and are in $\lhb$ relation, that is $\hbloc(a,b)$ holds. (3) $a$ has a different-location $\po$-successor $c$, and event $b$ has a different-location $\po$-predecessor $d$, and $(c,d)$ is in happens-before relation.

Based on the $\scb$ relation, RC11 defines $\pscb$ and $\pscf$. Relation $\pscb$ relates a pair of SC (memory access or fence) events and $\pscf$ relates a pair of SC fence events. 
Finally, RC11 defines $\psc$ relation by combining $\pscb$ and $\pscf$ relations. 

\noindent{\bf RC11 to $\srcmm$}
The $\srcmm$ model refines the RC11 relations with inclusion ($\incl$). Relation $\incl(a,b)$ holds when  
(i) $a$ and $b$ are atomic events, (ii) if the scope of $a$ or $b$ includes the thread of $b$ or $a$ respectively,
 and 
(iii) if both $a$ and $b$ access memory then they access the same memory location. Note that the $\incl$ relations are non-transitive, that is, $\incl(a,b)$ and $\incl(b,c)$ \emph{ does not} imply an $\incl(a,c)$ relation. 
To see this, consider events $a,b,c$ having scopes $\wg_1,\gpu_1$ and $\wg_2$ respectively where $\wg_1,\wg_2$ belong to GPU $\gpu_1$.
Then we have $\incl(a,b)$ and $\incl(b,c)$ but not $\incl(a,c)$.

Based on the $\incl$ relation, the $\rseq$, $\sw$, and $\lhb$ relations are extended in the $\srcmm$ model. 
In $\srcmm$, the $\rf$ relation in the $\rseq$ and $\sw$ relations must also be in the $\incl$ relation. Note that, even then, the $\sw$ related events may not be in the $\incl$ relation. 
Finally, $\lhb$ in $\srcmm$ is the transitive closure of the $\po$ and $\incl$-related $\sw$ relations. 

\noindent{\bf $\srcmm$ axioms} 
An execution in $\srcmm$ is consistent when it satisfies the axioms in Figure \ref{fig:scr11}.
The (Coherence) axiom ensures that the $\lhb$ relation or the combination of $\lhb$ and $\eco$ relations is irreflexive and does not create any cycle in the execution graph. 
The (Atomicity) axiom ensures that there is no intermediate event on the same memory location between a pair of events that are $\rmw$-related.
The \emph{SC} axiom forbids any cycle between the SC events which are both in the $\psc$ relation and the  $\incl$ relation.
Finally, the (No-Thin-Air) axiom forbids any cycle composed of $\po$ and $\rf$ relations. 
These axioms essentially forbid the patterns shown in \Cref{fig:scr11} in an execution graph. 
Among these scoped-RC11 axioms, (Atomicity) and (No-Thin-Air) are the same as those of RC11. The (Coherence) and (SC) axioms differ as they use more fine-grained $\incl$ relations for the scoped accesses.

\textbf{Example} Consider the program and its execution graphs in \Cref{fig:smp:srcmm}. 
If $i \neq j$, then the accesses on $Y$ do not synchronize, resulting in \Cref{fig:ex3:smp}. 
If $i=j$ then the accesses on $Y$ synchronize which results in \Cref{fig:ex2:smp}. 
The execution in \Cref{fig:ex5:smp} is forbidden as it violates the (Coherence) axiom.

\section{$\projname$ : Model Checking under $\srcmm$}
\label{sec:dpor}
In this section we discuss the $\projname$ approach in \Cref{subsec:dpor} followed by a running example in \Cref{subsec:door:example}. Finally, in \Cref{subsec:door:theorems} we discuss the soundness, completeness, and optimality of the proposed exploration algorithm. 
\subsection{DPOR Algorithm}
\label{subsec:dpor}
$\projname$ extends GenMC-TruSt and is in the same spirit as other well known dynamic partial order reduction (DPOR) algorithms \cite{DBLP:conf/popl/FlanaganG05,DBLP:journals/pacmpl/Kokologiannakis18,10.1145/2806886,10.1145/3276505,DBLP:journals/jacm/AbdullaAJS17,DBLP:conf/pldi/Kokologiannakis19,DBLP:journals/pacmpl/Kokologiannakis22}. 
\setlength{\intextsep}{1pt}
\begin{wrapfigure}[31]{r}{0.55\textwidth}
\begin{minipage}{0.56\textwidth}
\input{algo1}
\end{minipage}
\end{wrapfigure}
It verifies a program by exploring all its executions in a systematic manner, ensuring that no execution is visited more than once. Like \cite{DBLP:journals/pacmpl/Kokologiannakis22}, our algorithm also takes only 
polynomial space. 

\underline{Outline} \Cref{alg:alg1} invokes the $\explore$ procedure to explore the executions of input program under $\srcmm$. The $\explore$ procedure uses \Cref{alg:algprfs} to enable a read operation to read-from possible writes and thereby explore multiple executions, \Cref{alg:algisopt} to ensure no execution is explored more than once, and \Cref{alg:algracedetection} to identify and fix errors.   

\underline{$\explore$ procedure} The $\explore$ procedure explores executions $\G$, starting from an empty execution $\G_{\emptyset}$ where $E = \emptyset$, as long as they are consistent for a given memory model, in this case $\srcmm$ (see \Cref{incons1,incons2,incons3,incons4} of \Cref{alg:alg1}).
 Next, if some of the threads are  waiting
 at a barrier, while all other threads have finished execution,   
 then we observe a \emph{barrier divergence}, and the execution is said to be 
 $\mathsf{Blocked}$. 
 In a blocked execution, different threads may be waiting at different barriers.
 In this case (line \ref{divergence}), we report the divergence and terminate. Otherwise, we continue exploration by picking the next event (line \ref{nextevent}). This schedules 
 a thread and the next enabled event of that thread.
We use the total order $<_{exe}$ to denote the order in which events are added to the execution.

 The exploration stops if an assertion is violated (line \ref{violate}), or when all events from all threads are explored (line \ref{alldone}). The algorithm reports an error in the first case and in the second case outputs the graph $\G$.


If the exploration is not complete and the current event $e$ is a write (line \ref{case-write}), then 
the procedure $\checkrace$ detects races due to events conflicting with 
$e$ (line \ref{checkrace}), and also offers to repair them. 
On detecting a race, the algorithm chooses one of the following based on user choice --
(i) announce the race and stop exploration, or
(ii) announce the race and continue exploration, or (iii) announce the race and repair the race. 

Apart from calling  $\explore$ recursively (\Cref{exploreafter-co}) after adding the necessary $\co$ edges (line \ref{addCOs}) to $\G$, we check if 
$e$ can be paired with any existing read in $\G$ (line \ref{postponedrfs}). These reads are called ``reversible'' as we can reverse their order in the execution by placing them after the writes they read from. 
On a read event $r$, we consider all possible $\rf$s for $r$  and extend the execution $\G$ to a new execution $\G'$ (\addrf, \cref{addrf}).

\SetKwComment{Comment}{//}{}
\begin{algorithm}[t]
\DontPrintSemicolon
\SetKwFunction{FProc}{$\declarepostponed$ \label{alg:algpostponedrfs}}
\SetKwProg{Proc}{Procedure}{}{}
{
\textbf{let} \textbf{R} be set of \texttt{reversible} reads in $\egraph$ 
\label{reversible-reads}

\For{each $r=$$R(x,\_)$$ \in \textbf{R} \;\text{s.t.}\;  r \notin \po\rf.w $} 
{   
\label{iterate}
    $\deletedevents \leftarrow \{e \in \eventset \;|\; r 
    ~\gorder~ e \land  e \notin \po\rf.w \}$
    \label{erased}
    
    \Comment{\footnotesize{$\po\rf.w {=}\{e\mid {\exists} $ 
     {a} ${\po}, {\rf}$ path in {$\egraph$}  from $e$ to $w\}$ }}
        \If{$\isoptimal(\egraph,\deletedevents \cup \{r\},w,r)$} {
        \label{procisoptimal}
                    
        $\egraph' \leftarrow \addrf($$\egraph$$|_{E\backslash\deletedevents},w,r)$, 
        \label{remove-erased} 
        
        \textbf{for} each read $r \in $ $\egraph' \cap$ \textbf{R} $\cap$$\;\po\rf.w$ 
        \label{allrevereads}        
        set $\tt{reversible}$$(r)=False$
        \label{nomore-reverse}
        
        $\explore($$\pgm$$, $$\egraph'$$)$ 
        \label{exploreerase}
    }
}
}
\caption{ $\declarepostponed ($$\egraph$$ , w)$ \label{alg:algprfs}}
\end{algorithm}

\vspace{0.6em}
\begin{algorithm}[ht]
\DontPrintSemicolon
\SetKwFunction{FProc}{$\isoptimal$ \label{alg:algisoptimal}}
\SetKwProg{Proc}{Procedure}{}{}
{
    \DontPrintSemicolon
    \For{each event $e \in \deletedevents$}{
        \Comment{ \footnotesize{ $RF(e)$ is the write from which $e$ reads}}
        \lIf{$e=\rd(x) \land$ $ e ~\gorder~ RF(e) \land RF(e) \in \deletedevents$} 
        {\textbf{return} false}
        \label{eisread}
        {
            $e' \leftarrow$  \lIf{$e=\wt(x,v)$} {$e$ \textbf{else}{$~RF(e)$} }
            \label{eiswrite}

            \textbf{let} ${\tt{Eset}} = \{e'' \mid e'' ~\gorder~ e \lor e'' \in \po\rf.w \}$
            \label{Eset}
            
            \lIf{$e'~\co_{x}~ e''~\text{for some}~e'' \in \tt{Eset}$}
            {\label{line6check} \textbf{return} false}
            
             }
    }
    \textbf{return} true
}

\caption{ $\isoptimal(\egraph, \deletedevents, w, r)$  \label{alg:algisopt}}
\end{algorithm}

\underline{$\declarepostponed$ procedure}
The procedure pairs all reversible reads $r$ in $\G$ with all same-location write events $w$ (line \ref{reversible-reads}) provided $r$ is not in the  $\po \cup \rf$ prefix of $w$ in $\G$  (line \ref{iterate}), to preserve  
the (No-Thin-Air) axiom. Moreover, a new execution $\G'$ is obtained from $\G$ 
where $r$ reads from $w$ (line \ref{remove-erased}), and all events between $r$ and $w$ which are not $\po \cup \rf$ before $w$ are deleted (line \ref{erased}).

\underline{$\isoptimal$ procedure}
To ensure that no execution is explored twice, the  $\isoptimal$  procedure ensures that all writes in the deleted set are $\co$-maximal  wrt their location, and all reads 
in the deleted set read from $\co$-maximal writes. This is done by lines \ref{eisread} to \ref{line6check}
in  $\isoptimal$.

\underline{$\checkrace$ procedure} We check for races while adding each write $w$ to the execution.  For instance,  assume that all the reads and writes have been explored (\Cref{allwr}). For each event $e'$ in this set which is not related to $w$ by $\hab$, we check if any one of them is non-atomic to expose a \emph{data race}. If both have atomic accesses, we check if they are not scope-inclusive to report a \emph{heterogeneous race} (\Cref{iterate-reads}). Likewise, for each read event added, we consider all explored writes (line \ref{allreads}), and repeat the same check to expose a \emph{data race} or a \emph{heterogeneous race}. 

In addition, we also have an option of repair. In {\tt{Repair}} (line \ref{repair}, $\checkrace$), we either skip and return to $\explore$, or do the following repairs and terminate. 
First, if $e$ and $e'$ respectively have atomic and  non-atomic accesses  with non-inclusive scopes, then we update their scope to make them inclusive : for instance, if $e$,$e'$ are 
in different CTAs, we update their scopes to GPU-level.
Second, if  
at least one of $e, e'$ is a non-atomic access, then we update the 
non-atomic access to relaxed atomic, and update the scopes so that $e$, $e'$ have the same scope to prevent a heterogeneous race between them later. 
However, currently, we do not repair on non-atomic location data types. 
\vspace{0.6em}
\begin{algorithm}[ht]
\DontPrintSemicolon
\SetKwFunction{FProc}{$\checkrace$ \label{alg:algcheckrace}}
\SetKwProg{Proc}{Procedure}{}{}
{

\lIf {$e=W(x,v)$} {
$\mathsf{WR} \leftarrow$   set of seen reads/writes on $x$  
\label{allwr}
}
\lIf {$e=R(x,v)$}{
$\mathsf{WR} \leftarrow$ set of seen  writes on $x$ 
}
\label{allreads}

\For{each $e' \in \mathsf{WR}$ \text{s.t.}  $e' \notin \hab.e \wedge e \notin \hab.e'$} {   \label{iterate-reads} 
    \If{$\neg(\tt{IsAtomic}(e)) \lor \neg(\tt{IsAtomic}(e')) \lor \neg(\tt{IsScopeInclusive}(\egraph , e , e'))$}{
        $\tt{ReportRace(e,e')}$

        $\tt{Repair}(\pgm,\egraph,e,e')$ \label{repair}

    }
}
}
\caption{ $\checkrace(\egraph,e)$ \label{alg:algracedetection}}
\end{algorithm}



\noindent{\bf{Comparison with State-of-the-Art}}. We discuss how our algorithm differs from the existing DPOR algorithms.  The first 
departure comes in the $\explore$ procedure where we perform consistency checking: \Cref{incons1,incons2,incons3,incons4} 
are specific to the Scoped RC11 model which is not handled by any of the existing algorithms including the most recent \cite{DBLP:journals/pacmpl/Kokologiannakis22,abdulla2024parsimonious}, since none of them handle scoped models.   
The $\declarepostponed$ procedure is standard in all DPOR algorithms  and checks if we can pair reads with eligible writes  
which have been explored later. Next we have $\isoptimal$, which ensures that we are optimal while exploring executions: 
here, the optimality check of \cite{abdulla2024parsimonious} is tailored for sequential consistency; we extend 
the optimality checking algorithm for RC11 \cite{DBLP:journals/pacmpl/Kokologiannakis22} to $\srcmm$. While optimality is achieved  
by ensuring $\co$-maximality on writes \cite{DBLP:journals/pacmpl/Kokologiannakis22},  there could be optimal $\co$ orderings that 
are inconsistent in the non-scoped setting, which are consistent in the scoped case which need to be considered to 
achieve completeness. This needed careful handling to achieve polynomial space 
just as \cite{DBLP:journals/pacmpl/Kokologiannakis22}. Finally,  our $\checkrace$ algorithm is novel and 
differs from all existing approaches as it reports and also repairs heterogeneous races.
\subsection{Exploring the Executions of \ref{prg:dpor}}
\label{subsec:door:example}
We now illustrate the $\projname$ algorithm on program  \ref{prg:dpor} as a running example. The assertion violation to check is ${\bf {exists}} (a=1 \wedge b=1)$.  
This program has 4 consistent executions under $\srcmm$.

 \noindent  The exploration begins with the empty execution, with no events and all relations being empty.  As we proceed with the exploration, we use numbers $\circledsmall{1},\circledsmall{2},\dots$ to denote the order in which events are added to the execution. 
 Among the enabled events, we have the read from $Y$, namely, $a=Y_\MOna$ in thread $\T_2$ and the write to $X$ in $\T_1$. We add two events for these accesses to the execution (lines \ref{case-read}, \ref{addrf}, \ref{case-write} in $\explore$). 
The read on $Y$ has only the initial value 0 to read from; this is depicted by the $\rf$ edge to $\circledsmall{1}$, obtaining $\G_1$. On each new call to $\explore$, the partial execution is checked for consistency (lines \ref{incons1}-\ref{incons4}). $\G_1$ is consistent. 
\makebox[\textwidth]{
    \begin{minipage}{0.45\textwidth}
        \textcolor{black}{
        \[
        \mytag{SEG}
        \label{prg:dpor}
        \inarrC{
        X=Y=0;
        \\[1.2ex]
        \inparII{
        \T_1\tup{\wg_1, \_}
        \\[1.2ex]
        X^\wg_\MOrel = 1;
        \\[1.1ex]
        Y^\wg_\MOrel = 1;
        }{
        ~~\T_2\tup{\wg_1, \_}
        \\[1.2ex]
        ~~a=Y^\wg_\MOna;
        \\[1.1ex]
        ~~b=X^\wg_\MOacq;
        }
        }
        \]
        }
    \end{minipage}
    \hfill
    \begin{minipage}{0.47\textwidth}
    \centering
    \begin{tikzpicture}[yscale=0.6]
    {
      \node (s) at (2,0.2)  {$[init]$ };
      \node (s11) at (.8,-1)  {$\circledsmall{2}{:}\W^{\wg}_\MOrel(X,1)$};
          \node (ss) at (4,0.2)  {$\G_1$};
      \node (s21) at (3.2,-1)  {$\circledsmall{1}{:}\R^{\wg}_\MOna(Y,0)$};
      \node (s22) at (1.2,-2.5)  {};
    \draw[rf] (s) edge (s21); 
    \node[draw=olive, thick, dotted, fit=(s)(s11)(ss)(s21)(s22),inner sep=0.2pt] {};
    }
    \end{tikzpicture}  
    \end{minipage} 
}
Next, the read event on $X$ from $\T_2$ is added (line \ref{case-read}) having two sources to read from $X$ (line \ref{allw}): the initial write to $X$, and the write event $\circledsmall{2}$. This provides two branches to be explored, with consistent executions 
$\G_2$ and $\G_3$ respectively. 
 \makebox[\textwidth]{
     \begin{tikzpicture}[yscale=0.6]
     {
      \node (g2s) at (0,0)  {$[init]$ };
      \node (g2s11) at (-1.2,-1)  {$\circledsmall{2}{:}\W^{\wg}_\MOrel(X,1)$};
      \node (g2ss) at (2,0)  {$\G_2$};
            \node (g2s21) at (1.1,-1)  {$\circledsmall{1}{:}\R^{\wg}_\MOna(Y,0)$};
     \node (g2s22) at (1.2,-2.5)  {$\circledsmall{3}{:}\R^{\wg}_\MOacq(X,0)$};
            \draw[rf,bend right=21] (g2s) edge (g2s22);
      \draw[rf] (g2s) edge (g2s21);
      \node (g3s) at (6,0)  {$[init]$ };
      \node (g3s11) at (4.8,-1)  {$\circledsmall{2}{:}\W^{\wg}_\MOrel(X,1)$};
      \node (g3ss) at (8,0)  {$\G_3$};
            \node (g3s21) at (7.2,-1)  {$\circledsmall{1}{:}\R^{\wg}_\MOna(Y,0)$};
     \node (g3s22) at (7.2,-2.5)  {$\circledsmall{3}{:}\R^{\wg}_\MOacq(X,1)$};
      \draw[rf] (g3s11) edge (g3s22);
      \draw[rf] (g3s) edge (g3s21); 
       \node[draw=olive, thick, dotted, fit=(g2s)(g2s11)(g2ss)(g2s21)(g2s22)(g3s)(g3s11)(g3ss)(g3s21)(g3s22), minimum width=\textwidth ,inner sep=0.2pt] {};
       }
    \end{tikzpicture}   
 \label{fig:ex2:seg}    
}

Next, we add write on $Y$ from $\T_1$ to $\G_2, \G_3$ which results in executions $\G_7$ and $\G_4$ respectively. Both $\G_4$ and $\G_7$ are consistent executions. 

\makebox[\textwidth]{
     \begin{tikzpicture}[yscale=0.6]
     {
      \node (g7s) at (0,0)  {$[init]$ };
      \node (g7s11) at (-1.2,-1)  {$\circledsmall{2}{:}\W^\wg_\MOrel(X,1)$};
      \node (g7s21) at (1.2,-1)  {$\circledsmall{1}{:}\R^{\wg}_\MOna(Y,0)$};
     \node (g7s22) at (1.2,-2.5)  {$\circledsmall{3}{:}\R^\wg_\MOacq(X,0)$};
  \node (g7s12) at (-1.2,-2.5)  {$\circledsmall{4}{:}\W^\wg_\MOrel(Y,1)$};
      \draw[po] (g7s21) edge (g7s22);
  \draw[po] (g7s11) edge (g7s12);
    \node (g7ss) at (2,0)  {$\G_7$};
         \draw[rf] (g7s) edge (g7s21);
      \draw[rf,bend right=30] (g7s) edge (g7s22);
      \node (g4s) at (6,0)  {$[init]$ };
      \node (g4s11) at (4.8,-1)  {$\circledsmall{2}{:}\W^\wg_\MOrel(X,1)$};
      \node (g4s12) at (4.8,-2.5)  {$\circledsmall{4}{:}\W^\wg_\MOrel(Y,1)$};
      \node (g4s21) at (7.2,-1)  {$\circledsmall{1}{:}\R^{\wg}_\MOna(Y,0)$};
           \node (g4s22) at (7.2,-2.5)  {$\circledsmall{3}{:}\R^\wg_\MOacq(X,1)$};
      \draw[po] (g4s21) edge (g4s22);
      \node (g4ss) at (8,0)  {$\G_4$};
           \draw[po] (g4s11) edge (g4s12);
      \draw[rf] (g4s11) edge (g4s22);
      \draw[rf] (g4s) edge (g4s21); 
      \node[draw=olive, thick, dotted, fit=(g4s)(g4s11)(g4ss)(g4s21)(g4s22)(g4s12)(g7s)(g7s11)(g7ss)(g7s21)(g7s22)(g7s12), minimum width=\textwidth ,inner sep=0.2pt] {};
      }
    \end{tikzpicture}   
}
\noindent{\bf{Reversible reads}}.
 In $\G_4$, we observe that the read on $Y$ ($\circledsmall{1}$) can also read from the write $\circledsmall{4}$ which was added 
to the execution later. Enabling $\circledsmall{1}$ to read from $\circledsmall{4}$ involves swapping these two events so that the write 
happens before the corresponding read. 
Since $\circledsmall{2}$ is $\po$-before $\circledsmall{4}$, both of these events must take place 
before the read from $Y$ ($\circledsmall{1}$) for the $\rf$ to be enabled. The read from $X$ ($\circledsmall{3}$) however, 
has no dependence on the events in the first thread and happens after $\circledsmall{1}$. 
Therefore, we can 
delete (line \ref{erased} in $\declarepostponed$)
$\circledsmall{3}$, and add the read from $X$ later, after enabling the $\rf$ from $\circledsmall{4}$ to $\circledsmall{1}$ (line \ref{remove-erased} in $\declarepostponed$). 
The optimality check (line \ref{procisoptimal} in $\declarepostponed$) is passed in this case (see also the paragraph on optimality below) and we obtain execution $\G_5$. 
\setlength{\intextsep}{0.25pt}
\begin{wrapfigure}[6]{r}{0.39\textwidth}
\makebox[0.4\textwidth]{
     \begin{tikzpicture}[yscale=0.6]
     {
      \node (s) at (0,0)  {$[init]$ };
      \node (s11) at (-1.2,-1)  {$\circledsmall{2}{:}\W^\wg_\MOrel(X,1)$};
      \node (s12) at (-1.2,-2.5)  {$\circledsmall{4}{:}\W^\wg_\MOrel(Y,1)$};
                  \node (s21) at (1.2,-1)  {$\circledsmall{1}{:}\R^{\wg}_\MOna(Y,1 )$};
     \draw[po] (s11) edge (s12);
     \node (ss) at (2,0)  {$\G_5$};
       \draw[rf] (s12) edge (s21); 
       \node[draw=olive, thick, dotted, fit=(s)(s11)(ss)(s21)(s12),inner sep=0.2pt] {};
       }
    \end{tikzpicture}   
}
\end{wrapfigure}
We continue exploring from $\G_5$, adding the read on $X$ (\Cref{case-read} in $\explore$) from $\T_2$.  Here, $X$ may read from (\Cref{allw}, $\explore$) the initial write or $\circledsmall{2}$. This results in executions $\G_6$ and $\G_8$ which are both consistent. 

\makebox[\textwidth]{
 \vspace{0.5em}
     \begin{tikzpicture}[yscale=0.6]
     {
      \node (g8s) at (0,0)  {$[init]$ };
      \node (g8s11) at (-1.2,-1)  {$\circledsmall{2}{:}\W^\wg_\MOrel(X,1)$};
      \node (g8s12) at (-1.2,-2.5)  {$\circledsmall{4}{:}\W^\wg_\MOrel(Y,1)$};
      \node (g8s21) at (1.2,-1)  {$\circledsmall{1}{:}\R^{\wg}_\MOna(Y,1 )$};
     \node (g8s22) at (1.2,-2.5)  {$\circledsmall{5}{:}\R^\wg_\MOacq(X,0)$};
      \draw[po] (g8s21) edge (g8s22);
        \node (g8ss) at (2,0)  {$\G_8$};
     \draw[po] (g8s11) edge (g8s12);
      \draw[rf, bend right=22] (g8s) edge (g8s22);
      \draw[rf] (g8s12) edge (g8s21);
      
      \node (g6s) at (6,0)  {$[init]$ };
      \node (g6s11) at (4.8,-1)  {$\circledsmall{2}{:}\W^\wg_\MOrel(X,1)$};
      \node (g6s12) at (4.8,-2.5)  {$\circledsmall{4}{:}\W^\wg_\MOrel(Y,1)$};
      \node (g6s21) at (7.2,-1)  {$\circledsmall{1}{:}\R^{\wg}_\MOna(Y,1 )$};
     \node (g6s22) at (7.2,-2.5)  {$\circledsmall{5}{:}\R^\wg_\MOacq(X,1)$};
      \draw[po] (g6s21) edge (g6s22);
        \node (g6ss) at (8,0)  {$\G_6$};
      \draw[rf] (g6s11) edge (g6s22);
      \draw[rf] (g6s12) edge (g6s21); 
      \node[draw=olive, thick, dotted, fit=(g6s)(g6s11)(g6ss)(g6s21)(g6s22)(g6s12)(g8s)(g8s11)(g8ss)(g8s21)(g8s22)(g8s12), minimum width=\textwidth ,inner sep=0.2pt] {};
      }
    \end{tikzpicture}   
}

\noindent{\bf {Optimality}}. From $\G_7$, we do not consider the possibility of $Y$ reading from $\circledsmall{4}$ as it would result in an execution identical to $\G_8$, and consequently violate optimality. 
The $\isoptimal$ procedure checks it to ensure that no execution is explored more than once.  This check enforces a ``$\co$-maximality'' criterion on the events that are deleted while attempting a swap between a read event and a later write event: this is exactly where $\G_4$ and $\G_7$ differ. 
In $\G_7$, while considering the later write on $Y$ ($\circledsmall{4}$) to read from for the read event ($\circledsmall{1}$), the deleted (line \ref{erased}, $\declarepostponed$)
read event on $X$ ($\circledsmall{3}$) reads from the initial write of $X$ which is not $\co$-maximal since it is $\co$-dominated by $\circledsmall{2}$ (lines \ref{eiswrite}-\ref{line6check} in $\isoptimal$). Hence, the check-in line \ref{line6check} of $\isoptimal$ fails. 
In $\G_4$ however, the deleted read on $X$ ($\circledsmall{3}$)  reads from a $\co_x$-maximal write,  and the test passes. 
Thus, the algorithm only considers the possibility 
of the $Y$ reading from $\circledsmall{4}$ in $\G_4$, avoiding redundancy.


\noindent{\bf Program repair}
The exploration algorithm detects the assertion violation in $\G_6$ (since both $a,b$ read values 1) and detects a data race between $\circledsmall{1}$ and $\circledsmall{4}$. 




If $\projname$ exploration encounters a heterogeneous race between a pair of accesses then $\projname$ automatically repairs the race. To do so, $\projname$ changes the scope of the accesses to enforce an inclusion relation. 
After fixing a heterogeneous race $\projname$ terminates its exploration. 

Consider a variant of the \ref{prg:dpor} program where $T_1$ and $T_2$ are in different CTAs, $\projname$ fixes the heterogeneous race by transforming the scope from $\wg$ to $\gpu$. 
{
\[
\inarrC{
X=Y=0;
\\[1.2ex]
\inparII{
\T_1\tup{\wg_1, \_}
\\[1.2ex]
X^\wg_\MOrel = 1;
\\[1.1ex]
Y^\wg_\MOrel = 1;
}{
~~\T_2\tup{\wg_2, \_}
\\[1.2ex]
~~a=Y^\wg_\MOna;
\\[1.1ex]
~~b=X^\wg_\MOacq;
}
}
\ \leadsto \
\inarrC{
X=Y=0;
\\[1.2ex]
\inparII{
\T_1\tup{\wg_1, \_}
\\[1.2ex]
{\color{teal}X^\gpu_\MOrel = 1;}
\\[1.1ex]
Y^\wg_\MOrel = 1;
}{
~~\T_2\tup{\wg_2, \_}
\\[1.2ex]
~~a=Y^\wg_\MOna;
\\[1.1ex]
~~{\color{teal}b=X^\gpu_\MOacq;}
}
}
\]
}

\subsection{Soundness, Completeness and Optimality}
\label{subsec:door:theorems}
\begin{theorem}
    The DPOR algorithm for $\srcmm$ is sound, complete and optimal. 
\end{theorem}


\noindent{\bf{Soundness}}. The algorithm does not continue exploration from any inconsistent execution as ensured by \Cref{incons1,incons2,incons3,incons4} in \Cref{alg:alg1}, and is therefore sound.

\noindent{\bf{Completeness}}. The DPOR algorithm is complete as it does not miss any consistent and full execution. 
We prove this in the following steps:
\begin{myitem}
    \item We first show that starting from any consistent execution  $\egraph$,    we can uniquely
    roll back to obtain the  previous execution $\egraph_p$ (see the supplement for the algorithm to compute $\egraph_p$ from $\egraph$). 
    This is proved using  
    the fact that we have a fixed order in exploring the threads, along with the conditions that allow a swap between a read and a later write  
    to take place. To allow a swap of a read $r$ on some variable (say $x$), all events in $\deletedevents$
    respect ``$\co_x$-maximality''. This is enforced by $\isoptimal$  and allows us to uniquely construct the previous execution  $\egraph_p$.
    \item Second, we show that $\explore(\pgm,\egraph_p)$ leads to the call of $\explore(\pgm,\egraph)$.
    This shows that if $\egraph_p$ is reachable by the DPOR algorithm, then $\egraph$ is also reachable (see Supplement \ref{app:pred}, Lemma \ref{lem:reachblity}).
    \item In the final  step, we show that walking backward from any consistent $\egraph$ 
     we have a unique sequence  of executions $\egraph_p, \egraph_{p-1}, \egraph_{p-2}, \dots, $ 
    till  we obtain the empty execution $\egraph_{\emptyset}$. Thus, starting from $\explore(\pgm,\egraph_\emptyset)$, we obtain $\egraph$(
    Supplementary \ref{app:pred},  Lemma \ref{lem:finiteprev}, Lemma \ref{lem:prevseq}).
    
\end{myitem}
\noindent{\bf{Optimality}}. The algorithm is optimal as each full, consistent execution $\egraph$ is generated only once. Lines \ref{exploreafter-rf} and \ref{addCOs} of the $\explore$ procedure ensure that each recursive call to $\explore$ generates an execution that has a different $\rf$ edge or a different $\co$ edge. 
Also, during the $\declarepostponed$ procedure, 
the swap of a read  $r$ with a write  $w$ is successful only when 
the deleted events respect ``$\co_x$-maximality''.
As argued in completeness, for every (partial) consistent execution  $\egraph$, there exists a unique previous consistent execution $\egraph_p$. 

If the algorithm explores $\egraph$ twice, it means that there are two different exploration sequences with respective previous executions  $\egraph_{p}$ and $\egraph_{q}$.
This is a contradiction as we have a unique previous execution (see Supplementary \Cref{app:optmal}, \Cref{thm:optml}).  

\noindent{\bf Polynomial Space}
The DPOR algorithm explores executions recursively in a depth-first manner, with each branch explored independently. 
Since the recursion depth is bounded by the size of the program, this approach ensures that the algorithm uses only polynomial space.

\subsection{Exploring the Reads-From Equivalence}
For simplicity, we have focused our presentation on exploring executions that contain the $\co$ relation explicitly.
However, \cref{alg:alg1} can be easily adapted to explore executions where $\co$ is not given explicitly.
This corresponds to exploring the reads-from partitioning~\cite{Chalupa2017}, a setting that is also supported by GenMC~\cite{DBLP:journals/pacmpl/Kokologiannakis22}.
This is often a desirable approach, because it may significantly reduce the search space:~there can be exponentially many executions, differing only in their $\co$, all of which collapse to a single class of the reads-from partitioning.

Exploring the reads-from partitioning requires that every time a new execution is explored, the algorithm performs a consistency check to derive a $\co$, so as to guarantee that the execution is consistent.
If the program has no SC accesses, this check is known to be efficient for RC11~\cite{Abdulla:2018,Lahav:2015}, taking essentially linear time~\cite{tunc_optimal_2023}.
These results easily extend to scoped RC11, by adapting the computation of the happens-before relation so as to take the scope inclusion $\incl$ into consideration.
On the other hand, the presence of SC accesses makes the problem intractable~\cite{Gibbons:1997,Mathur:lics2020}, though it remains in polynomial time with a bounded number of threads~\cite{Gibbons:1997,Abdulla:2019b}.


\section{Experimental Evaluation}
\label{sec:evaluation}
We implement our approach as a tool (GPU Model Checker $\projname$)  capable of handling programs with scopes. 
$\projname$ is implemented in GenMC-Trust~\cite{DBLP:journals/pacmpl/Kokologiannakis22},
and takes scoped C/C++ programs as input and works at the LLVM IR level. 
Similar to existing approaches, we handle programs with loops by unrolling them by a user-specified number of times. 
We conduct all our experiments on an Ubuntu 22.04.1 LTS with Intel Core i7-1255U×12 and 16 GiB RAM.

We experiment with $\projname$ on a wide variety of programs starting from litmus tests to 
larger benchmarks. 
We mainly compare its performance with $\dartagnan$~\cite{dat3m,Dartagnan2021}, a state-of-the-art bounded model checker, which also handles programs with scope \cite{hernan24}. 
$\dartagnan$ has recently integrated the PTX and Vulkan GPU consistency models into its test suite.
Even though the consistency model considered by $\dartagnan$ are different from  
$\srcmm$, which $\projname$ considers, $\dartagnan$ is closest available tool to the kind of work we report in this paper.  
Two other tools that also handle programs with scopes are $\iguard$~\cite{iguard} and $\scordd$~\cite{Kamath:2020}.  
However, these tools do not reason about weak memory concurrency in GPUs. which makes their benchmarks not directly usable by $\projname$. 
In order to still experiment with them, we change their shared accesses to atomics.

\subsection{Comparison with $\dartagnan$}
\begin{table}[ht]
\centering
\begin{tabular}{|l c c| c c c| c c c|}
\hline
 &  &  & \multicolumn{2}{c}{\qquad \underline{$\dartagnan$}} &
& \multicolumn{2}{c}{\qquad \underline{$\projname$}} & \\
\textbf{Program} & \textbf{Grid} & \textbf{Threads} &  \textbf{Result} &  \textbf{Time} & \textbf{Memory} & \textbf{Result} &  \textbf{Time} & \textbf{Memory} \\
\hline
caslock      & 4,2  &  8     & NR & 1300 & 494   & NR   & 50    & 85  \\ 
caslock1    & 4,2  &  8    & R  & 0.7  &  304  & R    & 0.1    & 85 \\ 
caslock1    & 6,4  &  24    & R  & 2.5  & 670   & R    & 0.1   & 85  \\ 
caslock2    & 4,2  &  8    & R  & 0.6  &  270   & R    & 0.1   & 85 \\ 
caslock2    & 6,4  &  24    & R  & 2.3  & 680   & R    & 0.1   & 85 \\ 
\hline
ticketlock   & 4,2  &  8    & -  & TO   &  1062  & NR   & 320   & 85 \\ 
ticketlock1 & 4,2  &  8    & R  & 0.7  &  340  & R    & 0.1    & 84 \\ 
ticketlock1 & 6,4  &  24    & R  & 965  & 941   & R    & 0.1   & 84 \\ 
ticketlock2 & 4,2  &  8    & R  & 0.9  &  290  & R    & 0.1    & 84 \\ 
ticketlock2 & 6,4  &  24    & R  & 1020 & 952  & R    & 0.1    & 84 \\ 
\hline
ttaslock     & 3,2  &  6    & -  & TO   &  1116  & NR   & 500   & 84 \\ 
ttaslock1   & 4,2  &  8    & R  & 0.7  &  285  & R    & 0.1    & 84 \\ 
ttaslock1   & 6,4  &  24    & R  & 3.6  & 321   & R    & 0.1   & 84  \\ 
ttaslock2   & 4,2  &  8    &   R  & 0.7 & 324   & R    & 0.1   & 84 \\ 
ttaslock2   & 6,4  &  24   & R  & 4    &  917  & R    & 0.1    & 84 \\ 
\hline
XF-Barrier   & 4,3  &  12    & NR & 29   & 4200  & NR   & 28    & 85 \\ 
XF-Barrier1  & 4,3  &  12    & R  & 4    & 1380   & R    & 0.1  & 85  \\ 
XF-Barrier1  & 6,4  &  24    & NR*& 190  & 1476   & R    & 0.2  & 85 \\ 
XF-Barrier2  & 4,3  &  12    & R  & 9    & 1399   & R    & 0.1  & 85  \\ 
XF-Barrier2  & 6,4  &  24    & NR*& 170  & 1505   & R    & 0.2  & 85 \\ 
\hline
\end{tabular}
\caption{Data race detection: Evaluating on parameterized, single kernel code \label{tab:param}.
Time Out (TO) =30 minutes. (Time in Seconds and Memory in MB respectively).
The number of events per execution is less than 120.
In column Result, R denotes race detected and NR denotes no race. 
The * on two NR entries shows a wrong result in $\dartagnan$.
In Grid column, X,Y represent X CTAs and Y threads per CTA.}
\end{table}

We compare the performance of $\projname$ with $\dartagnan$~\cite{dat3m} on the implementation of four synchronization primitives (caslock, ticketlock, ttaslock, and XF-Barrier), taken from~\cite{dat3m,xiao2010interxfbar}. 
These benchmarks use relaxed atomics, which is a very important feature of real GPU APIs.
All the 1 (caslock1, ticketlock1, ttaslock1, and XF-Barrier1) and 2 (caslock2, ticketlock2, ttaslock2, and XF-Barrier2) variants are obtained by transforming the release and acquire accesses to relaxed accesses, respectively. 
Moreover, the XF-Barrier benchmark uses CTA-level barriers for synchronization.
\Cref{tab:param} shows the results of the evaluation of these applications. 
We parameterize these applications by increasing 
the number of threads in the program,  the number of CTAs, and the number of threads in a CTA. For comparing with $\dartagnan$, we focus on race detection. 

In \Cref{tab:param} the Grid and Threads columns denote the thread organization, 
and the total number of threads respectively. The Result column shows the observed result -- whether a race was detected (R), or whether the program was declared safe and no race was reported (NR).  
The Time and Memory columns show the
time taken in seconds and the memory consumed in MB taken by $\dartagnan$ and $\projname$.  

We observe that in all examples except  XF-Barrier, 
$\projname$ and $\dartagnan$ produce the same results, and $\projname$ outperforms 
$\dartagnan$ significantly in time and memory requirements. 
For the benchmarks XF-Barrier1 and XF-Barrier2 with grid structure (6,4) respectively, $\projname$ successfully detects the underlying data race within a fraction of a second. 
The time and memory requirements we have reported 
for $\dartagnan$ is with loop bound 12 as $\dartagnan$ is unable to find the race even after unrolling 
to loop bound 12. 
On increasing the loop bound to 13, 
$\dartagnan$ kills the process after showing a heap space error.
In conclusion, in all the benchmarks in \Cref{tab:param}, $\projname$ significantly outperforms $\dartagnan$.

\subsection{Verification of GPU Applications}
We evaluate $\projname$ on medium to large real GPU applications, particularly for 
heterogeneous race and barrier divergence errors.
\begin{table}[ht]
\vspace{1em}
\centering
\centering
\begin{tabular}[t]{|lccccc|}
\hline
\textbf{Program}  & \textbf{Grid} & \textbf{Threads} & \textbf{Events} & \textbf{Memory} & \textbf{Time} \\
\hline

1dconv12 & 12,4         & 48      &  1135     & 85   & 8.9 \\
1dconv15 & 15,4         & 60      &  1359     & 85  & 15.2 \\
1dconv20 & 20,4         & 80      &  1662     & 85   & 28.7 \\
1dconv25 & 25,4         & 100     &  1937     & 85   & 47.7 \\
\hline
GCON4 & 4,2        &8        &  493      & 126  & 2 \\
GCON5 & 5,2        &10       &  563      & 150  & 5 \\
GCON7 & 7,2        &14       &  697      & 176  & 25 \\
GCON10 & 10,2      & 20      &  901      & 250  & 75 \\
GCON15 & 15,2      & 30      &  1241     & 383  & 295 \\
\hline
GCOL4  & 4,2       & 8       & 337      & 85  & 0.5 \\
GCOL5  & 5,2       & 10      & 435      & 85  & 1.7 \\
GCOL7  & 7,2       & 14      & 643       & 86  & 3.6\\
GCOL10 & 10,2      & 20      & 1000      & 85 & 14.5\\
GCOL15 & 15,2      & 30      & 1051   & 88       & 18  \\
\hline
matmul4  & 4,3       & 12      & 1036      & 85       & 8   \\
matmul5  & 5,3      & 15      & 1054      & 84    & 9   \\
matmul7  & 7,3      & 21      & 1424       & 85  & 32  \\
matmul10 & 10,3      & 30      & 2556        & 125    & 360 \\
matmul15 & 15,3      & 45      & 2154       & 90        & 175 \\
\hline
\end{tabular}
\caption{Heterogenous race detection using $\projname$ on GPU Applications \label{tab:app}. (Time in Seconds and Memory in MB respectively). Events column represents the maximum number of events across all executions. }
\end{table}

\noindent{\bf Heterogeneous Races}
We experiment with four GPU applications -- OneDimensional
Convolution (1dconv), Graph Connectivity (GCON), Matrix Multiplication (matmul) and Graph
Colouring (GCOL) from~\cite{Kamath:2020,iguard}. 
Each program has about 250 lines of code.
For our experiments, we transform the accesses in these benchmarks with SC memory order and $\gpu$ scope. 
Finally, all these transformed benchmarks have SC accesses except GCON which has only relaxed accesses.
We do not execute $\dartagnan$ on these programs, as they are multikernel and involve CPU-side code, which makes it unclear how to encode them in $\dartagnan$.

\cref{tab:app} shows the 4 variants of each program by varying the grid structure. For instance, ldconv12 represents the version having 12 CTAs. 
The last two columns show the time and memory taken by $\projname$ in detecting the first heterogeneous race. 
The detection of the first heterogeneous races in the ldconv, GCON, GCOL, and matmul benchmarks takes 4, 455, 11, and 18 executions respectively. 
In all cases, $\projname$ detects the first race within 6 minutes. 

\begin{table}[ht]
\centering
\begin{tabular}[t]{|lccccc|}
\hline
\textbf{Program}  & \textbf{Grid} & \textbf{Threads} &  \textbf{Events} &  \textbf{Memory} &  \textbf{Time} \\
\hline
histogram4 & 4,2 & 8 &   144 &  85  & 0.1 \\
histogram6 & 6,4 & 24 &  104 &      85  & 0.1  \\
XF-Barrier4 & 4,2 & 8 & 132 &     85 & 0.1 \\
XF-Barrier6  & 6,4 & 24 & 369 &     85 & 0.7 \\
arrayfire-sm  & 1,16 & 16 & 1400 & 88 & 13 \\
arrayfire-wr  & 1,256 & 256 & 240 & 85 & 0.3 \\
GkleeTests1  & 2,32 & 64 & 700 & 85 & 2.5 \\
GkleeTests2  & 1,64 & 64 & 900 & 86 & 4 \\
\hline
\end{tabular}
\caption{ Barrier Divergence using $\projname$ on various grid-structured programs (Time in seconds, Memory in MB). Events column represents the maximum number of events seen across executions. \label{tab:barr-div}}
\end{table}
\begin{table}[ht]
\centering
\begin{tabular}{|lcccccc|}
\hline
\textbf{Program}  & \textbf{Grid} & \textbf{Threads} & \textbf{Executions} &  \textbf{Events} &  \textbf{\#Race} &  \textbf{\#Fix}  \\
\hline
bench1 & 2,3 & 6 &   720 &      77  & 1 & 2 \\
bench2 & 2,3 & 6 &   205236 &    83    & 2 & 4 \\
bench3 & 8,1 & 8 &   12257280 &    100  & 2 & 4 \\
bench4 & 4,2 & 8 &   12257280 &    100    & 2 & 4 \\
bench5 & 5,1 & 5 &   1200 &    65    & 3 & 3 \\
GCOL & 2,1 & 2 &  350242 &      459 & 3 &6  \\
matmul & 3,1 & 3 & 2409 &     1153 & 2 &1 \\
1dconv  & 2,2 & 4 & 995328  &    361 & 1 & 1 \\
\hline
\end{tabular}
\caption{Race Repair using $\projname$ on various  grid-structured programs. \#Race denotes the number of races detected and 
\#Fix represents the number of changes made to fix the race. 
Events column represents the maximum number of events seen across executions.
\label{tab:app-repair}}
\end{table}

\noindent{\bf Barrier Divergence}
Next, we evaluate $\projname$ for detecting barrier divergence, with the results shown in \cref{tab:barr-div}. 
We consider four GPU applications -- histogram \cite{syclbook}, XF-Barrier, arrayfire:select-matches (arrayfire-sm) and arrayfire:warp-reduce(arrayfire-wr) \cite{arrayfire,simulee}, as well as GkleeTests1 and GkleeTests2 kernels from the GKLEE tests~\cite{gklee,simulee}. 
All these benchmarks except Histogram use SC accesses and have barrier divergence. 
Histogram has a mix of SC and relaxed accesses. 
In our experiments, we introduce a barrier divergence bug in the original histogram program~\cite[Chapter 19]{syclbook}. 
We vary the grid structures, similar to the benchmarks created for experimenting with the heterogeneous race detection. 

\subsection{Race Repair}
Apart from detecting, $\projname$  also repairs heterogeneous races as shown in \cref{tab:app-repair} on five micro-benchmarks and three GPU applications~\cite{Kamath:2020,iguard}. 
The {\#Race} column shows the number of races detected and fixed and the {\#Fix} column shows the number of lines of code changes required to fix the detected races. 
In all cases, $\projname$ detects and repairs all races within 3 secs. 
After repair, we let $\projname$ exhaustively explore all executions of corrected programs (bench1, bench2, bench5, matmul finish within 10 minutes and bench3, bench4, GCOL and 1dconv finish within 6 hours). 
Finally, the Executions column shows the number of executions explored on running the corrected program, and the Events column shows the maximum number of events for all explored executions post-repair. 
\begin{table}[ht]
\vspace{0.5em}
\centering
\begin{tabular}[t]{|lcccc|}
\hline
\textbf{Program} & \textbf{Events} &  \textbf{Memory} & \textbf{Executions} &  \textbf{Time} \\
\hline
LB-3    & 36   & 84    & 7        & 0.3       \\
LB-7    & 76   & 84    & 127      & 0.4       \\
LB-12   & 126  & 84    & 4095     & 1.3       \\
LB-18   & 186  & 101   & 262143   & 228     \\
LB-22   & 226  & 127   & 4194303  & 5647  \\
\hline
\end{tabular}
 \caption{Scalability of $\projname$ on safe benchmark LB  (Time in Seconds and Memory in MB). 
 Executions column represents the executions explored.     
 \label{tab:lb-scalability}}
\end{table}

\subsection{Scalability}
\cref{fig:scalability} shows the scalability of $\projname$ for increasing number of threads on three benchmarks -- SB (store buffer) and two GPU applications 1dconv, GCON. For SB, we create 24 programs with increasing threads from 2 to 25. 
For 1dconv, we create 30 programs with increasing CTAs from 1 to 30 with four threads per CTA. 
For GCON, we create 50 programs with increasing threads from 1 to 50.
\cref{fig:scalability} shows the $\projname$ execution time and the memory consumed to detect the heterogeneous race for 1dconv and GCON and the assertion violation in SB; 
the x-axis shows the total number of threads for GCON, SB and CTAs for 1dconv, and the y-axis measures the memory in megabytes (MB) and the time in seconds.
We also experiment on the LB (load buffer) benchmark in \Cref{tab:lb-scalability}. 
We create 21 programs with increasing threads (LB-2 to LB-22) and exhaustively explore all consistent executions. 
We observe that in all benchmarks $\projname$ exhaustively explores more than 4 million executions within 5500 seconds.

\begin{figure}[tt]
\def\plotwidthone{13cm}
\def\plotwidththree{15cm}
\def\plotwidthtwo{7.5cm}
\def\plotheight{4.5cm}
\def\plotheighttwo{3.5cm}
\centering
\begin{subfigure}[t]{0.45\textwidth}
\scalebox{0.8}{
    \begin{tikzpicture}[tight background, inner sep=2pt,]
        \begin{axis}[inner sep=2pt,
		title={\large \underline{Time}},
            xlabel={\large Threads$/$CTAs},
            enlarge y limits={abs=0.6cm,upper},
		enlarge x limits={abs=8pt},
            xticklabel style={font=\large,rotate=50,anchor=east},
            yticklabel style={font=\large},
            ylabel style={font=\large},
            xtick={1, 5, 10, 15, 20, 25, 30, 35, 40, 45, 50},
            xticklabels={01,05,10,15,20,25,30,35,40,45,50},
            xtick pos=left,
            ylabel near ticks,
             width=\plotwidthtwo,
            height=\plotheighttwo,
            ymajorgrids=true,
            grid style=dashed,
            legend to name=nametime,
            legend style={font=\footnotesize, draw=black, legend columns=-1,/tikz/every even column/.append style={column sep=0.2cm}},
            grid=major
        ]
        \addplot[line width=1.3pt,color=brown] table [x expr=\coordindex, y index=1, col sep=space] {graph/1dconv_time.dat};
        \addlegendentry{1dconv}

        \addplot[line width=1.3pt,color=blue] table [x expr=\coordindex,y index=1, col sep=space] {graph/gcon_1_time.dat};
        \addlegendentry{GCON}

        \addplot[line width=1.3pt,color=red] table [x expr=\coordindex, y index=1, col sep=space] {graph/SB_time.dat};
        \addlegendentry{SB}
        
        \end{axis}
    \end{tikzpicture}
}
\end{subfigure}
\hspace*{\fill}
\begin{subfigure}[t]{0.5\textwidth}
\scalebox{0.8}{
    \begin{tikzpicture}[tight background, inner sep=2pt,]
        \begin{axis}[inner sep=2pt,
		title={\large \underline{Memory}},
            xlabel={\large Threads$/$CTAs},
            enlarge y limits={abs=0.6cm,upper},
		enlarge x limits={abs=8pt},
            xticklabel style={font=\large,rotate=50,anchor=east},
            yticklabel style={font=\large},
            ylabel style={font=\large},
            xtick={1, 5, 10, 15, 20, 25, 30, 35, 40, 45, 50},
            xticklabels={01,05,10,15,20,25,30,35,40,45,50},
            xtick pos=left,
            ylabel near ticks,
            width=\plotwidthtwo,
            height=\plotheighttwo,
            ymajorgrids=true,
            grid style=dashed,
            legend to name=namegmem,
            legend style={font=\footnotesize, draw=black, legend columns=-1,/tikz/every even column/.append style={column sep=0.2cm}},
            grid=major
        ]

        \addplot[line width=1.3pt,color=brown] table [x expr=\coordindex, y index=1, col sep=space] {graph/1dconv_mem.dat};
        \addlegendentry{1dconv}
        
        \addplot[line width=1.3pt,color=blue] table [x expr=\coordindex, y index=1, col sep=space] {graph/gcon_1_mem.dat};
        \addlegendentry{GCON}

        \addplot[line width=1.3pt,color=red] table [x expr=\coordindex, y index=1, col sep=space] {graph/SB_mem.dat};
        \addlegendentry{SB}
        
        \end{axis}
    \end{tikzpicture}
}
\end{subfigure}
    \caption{ {\tiny{\ref{namegmem}}} Scalability: $\projname$ execution time and the memory consumed to detect the heterogeneous race for 1dconv and GCON and the assertion violation in SB. The x-axis shows the total number of threads for GCON and SB, and CTAs for 1dconv. The y-axis measures the memory in megabytes (MB) and the time in seconds.}
    \label{fig:scalability}
\end{figure}
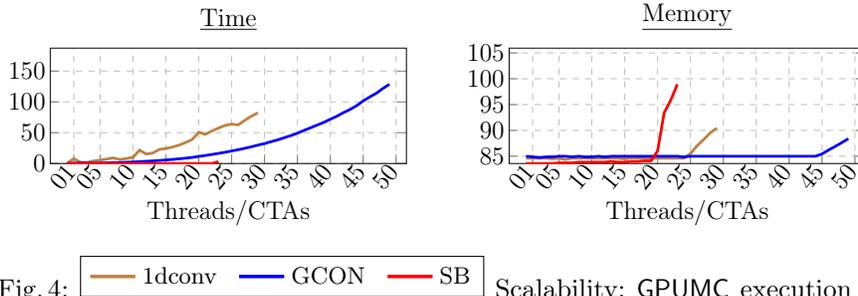

\section{Related Work}
\label{sec:related}

\noindent \textit{Semantics} 
Weak memory concurrency is widely explored in programming languages for CPUs and GPUs \cite{Batty:2011,Lahav:2017,Kang:2017,Chakraborty:2019}, 
\cite{Batty:2016,Orr:2015,hower2014heterogeneous,vulkan,Sorensen:2018}, compilers \cite{Podkopaev:2019,Chakraborty:2017}, and CPU and GPU architectures \cite{Alglave:2021,Pulte:2018,Alglave:2015,Lustig:2019,Lustig:2022,Goens:2023}. Although $\projname$ follows scoped-RC11 semantics \cite{Lustig:2019}, it is possible to adapt our approach to several other GPU semantic models. 
However, developing a DPOR model checker for GPUs with all guarantees that explore executions with $\po \cup \rf$ cycle is a nontrivial problem, in general  \cite{Kang:2017,Chakraborty:2019,Jeffrey:2016,10.1145/3563315}, which is future work. 

\noindent\textit{GPU testing}
Testing of GPU litmus programs is used to reason about GPU features \cite{SorensenD16,Alglave:2015,Sorensen:2021}, reveal errors \cite{SorensenD16}, weak memory behaviors \cite{Alglave:2015}, and various progress properties \cite{Sorensen:2021,Sorensen:oopsla:2016,Sorensen:2018,Ketema:2017}. 
Complementarily, our model checker explores all executions to check the correctness of the GPU weak-memory programs.

\noindent\textit{Verification and testing of weak memory concurrency}
There are several DPOR algorithms for the verification of shared memory programs under weak memory such as TSO, PSO,  release-acquire (RA) and RC11
\cite{Bui2021,DBLP:journals/pacmpl/Kokologiannakis18,10.1145/2806886,10.1145/3276505,DBLP:journals/jacm/AbdullaAJS17,10.1007/978-3-662-46681-0_28,10.1145/2737924.2737956}. DPOR algorithms have also been developed for weak consistency models such as 
CC, CCv and CM \cite{DBLP:conf/tacas/AbdullaAKGT23}. 
These are sound, complete, and optimal, although 
they incur an exponential memory usage. Recently, \cite{DBLP:conf/pldi/Kokologiannakis19,DBLP:journals/pacmpl/Kokologiannakis22} proposed a DPOR algorithm applicable to a range of weak memory models incurring only polynomial space while being also sound, complete and optimal. On the testing front, we have 
tools such as C11-tester \cite{Luo:2021}, and tsan11rec \cite{Lidbury:2017} for several variants of C11 concurrency. However, these tools do not address the verification of programs with scopes. 

\noindent\textit{GPU analysis and verification} 
Several tools propose analysis and verification of GPU programs including $\gpuverify$ \cite{Betts:2012}, $\gklee$ \cite{gklee}, GPUDrano \cite{Alur:2017}, $\scordd$ \cite{Kamath:2020}, $\iguard$ \cite{iguard}, $\simulee$ \cite{simulee}, $\textsc{SESA}$ \cite{sesa} for checking data races \cite{iguard,Kamath:2020,Eizenberg:2017,Betts:2012,Zheng:2014,Peng:2018,haccrg,li2017ld,grace}, divergence \cite{Coutinho:2013,Eizenberg:2017,Betts:2012}. Other relevant GPU tools are $\textsc{PUG}$ \cite{pug1,pug2} 
and $\textsc{Faial}$ \cite{Faial}.
 However, these do not handle weak memory models.  

\section{Conclusion}
\label{sec:conclusion}
We present $\projname$, a stateless model checker developed on theories of DPOR for GPU weak memory concurrency, which is sound, complete, and optimal, and uses polynomial space. 
$\projname$ scales to several larger benchmarks and applications, detects errors, and automatically fixes them. 
We compare $\projname$ with state-of-the-art tool \dartagnan, a bounded model checker for GPU weak memory concurrency. Our experiments on $\dartagnan$ benchmarks reveal errors that remained unidentified by \dartagnan. 


\subsubsection*{Acknowledgements.}
This work was partially supported by a research grant (VIL42117) from VILLUM FONDEN.

\pagebreak
\bibliographystyle{splncs04}
\bibliography{references}

\begin{thebibliography}{10}
\providecommand{\url}[1]{\texttt{#1}}
\providecommand{\urlprefix}{URL }
\providecommand{\doi}[1]{https://doi.org/#1}

\bibitem{vulkan}
Vulkan memory model. \url{https://github.com/KhronosGroup/Vulkan-MemoryModel}

\bibitem{NVIDIA_CUDA_v12.6}
Cuda c++ programming guide (2024), available at:
  [https://docs.nvidia.com/cuda/cuda-c-programming-guide/index.html]

\bibitem{cccl}
Cuda core compute libraries (cccl) (2024), available at:
  [https://github.com/NVIDIA/cccl]

\bibitem{cutlass}
Cutlass 3.6.0 (2024), available at: [https://github.com/NVIDIA/cutlass]

\bibitem{opencl}
The opencl™ specification (2024), available at:
  [https://registry.khronos.org/OpenCL/specs/3.0-unified/html/OpenCL\_API.html]

\bibitem{10.1007/978-3-662-46681-0_28}
Abdulla, P.A., Aronis, S., Atig, M.F., Jonsson, B., Leonardsson, C., Sagonas,
  K.: Stateless model checking for tso and pso. In: Baier, C., Tinelli, C.
  (eds.) Tools and Algorithms for the Construction and Analysis of Systems. pp.
  353--367. Springer Berlin Heidelberg, Berlin, Heidelberg (2015)

\bibitem{DBLP:journals/jacm/AbdullaAJS17}
Abdulla, P.A., Aronis, S., Jonsson, B., Sagonas, K.: Source sets: {A}
  foundation for optimal dynamic partial order reduction. J. {ACM}
  \textbf{64}(4),  25:1--25:49 (2017). \doi{10.1145/3073408},
  \url{https://doi.org/10.1145/3073408}

\bibitem{abdulla2024parsimonious}
Abdulla, P.A., Atig, M.F., Das, S., Jonsson, B., Sagonas, K.: Parsimonious
  optimal dynamic partial order reduction. In: International Conference on
  Computer Aided Verification. pp. 19--43. Springer (2024)

\bibitem{Abdulla:2019b}
Abdulla, P.A., Atig, M.F., Jonsson, B., L{\aa}ng, M., Ngo, T.P., Sagonas, K.:
  Optimal stateless model checking for reads-from equivalence under sequential
  consistency. Proc. {ACM} Program. Lang.  \textbf{3}({OOPSLA}),  150:1--150:29
  (2019). \doi{10.1145/3360576}, \url{https://doi.org/10.1145/3360576}

\bibitem{10.1145/3276505}
Abdulla, P.A., Atig, M.F., Jonsson, B., Ngo, T.P.: Optimal stateless model
  checking under the release-acquire semantics. Proc. ACM Program. Lang.
  \textbf{2}(OOPSLA) (Oct 2018). \doi{10.1145/3276505},
  \url{https://doi.org/10.1145/3276505}

\bibitem{Abdulla:2018}
Abdulla, P.A., Atig, M.F., Jonsson, B., Ngo, T.P.: Optimal stateless model
  checking under the release-acquire semantics. Proc. ACM Program. Lang.
  \textbf{2}(OOPSLA) (oct 2018). \doi{10.1145/3276505}

\bibitem{DBLP:conf/tacas/AbdullaAKGT23}
Abdulla, P.A., Atig, M.F., Krishna, S., Gupta, A., Tuppe, O.: Optimal stateless
  model checking for causal consistency. In: Proceedings of TACAS 2023. vol.
  13993, pp. 105--125. Springer (2023),
  \url{https://doi.org/10.1007/978-3-031-30823-9\_6}

\bibitem{Alglave:2015}
Alglave, J., Batty, M., Donaldson, A.F., Gopalakrishnan, G., Ketema, J.,
  Poetzl, D., Sorensen, T., Wickerson, J.: {GPU} concurrency: Weak behaviours
  and programming assumptions. In: PASPLOS 2015. pp. 577--591.
  \doi{10.1145/2694344.2694391}

\bibitem{Alglave:2021}
Alglave, J., Deacon, W., Grisenthwaite, R., Hacquard, A., Maranget, L.: Armed
  cats: Formal concurrency modelling at arm. ACM Trans. Program. Lang. Syst.
  \textbf{43}(2) (2021). \doi{10.1145/3458926}

\bibitem{Alur:2017}
Alur, R., Devietti, J., Navarro~Leija, O.S., Singhania, N.: Gpudrano: Detecting
  uncoalesced accesses in gpu programs. In: Majumdar, R., Kun{\v{c}}ak, V.
  (eds.) Computer Aided Verification. pp. 507--525. Springer International
  Publishing, Cham (2017)

\bibitem{Batty:2016}
Batty, M., Donaldson, A.F., Wickerson, J.: Overhauling {SC} atomics in {C11}
  and {OpenCL}. In: POPL '16. pp. 634--648. ACM (2016).
  \doi{10.1145/2837614.2837637}

\bibitem{Batty:2011}
Batty, M., Owens, S., Sarkar, S., Sewell, P., Weber, T.: Mathematizing {C++}
  concurrency. In: POPL'11. pp. 55--66. {ACM} (2011).
  \doi{10.1145/1926385.1926394}

\bibitem{Betts:2012}
Betts, A., Chong, N., Donaldson, A.F., Qadeer, S., Thomson, P.: Gpuverify: a
  verifier for {GPU} kernels. In: Leavens, G.T., Dwyer, M.B. (eds.) OOPSLA
  2012. pp. 113--132 (2012). \doi{10.1145/2384616.2384625}

\bibitem{Bui2021}
Bui, T.L., Chatterjee, K., Gautam, T., Pavlogiannis, A., Toman, V.: The
  reads-from equivalence for the {{TSO}} and {{PSO}} memory models. Proceedings
  of the ACM on Programming Languages  \textbf{5}(OOPSLA),  1--30 (Oct 2021).
  \doi{10.1145/3485541}, \url{https://dl.acm.org/doi/10.1145/3485541}

\bibitem{Chakraborty:2017}
Chakraborty, S., Vafeiadis, V.: Formalizing the concurrency semantics of an
  llvm fragment. In: CGO 2017. p. 100–110. CGO '17 (2017)

\bibitem{Chakraborty:2019}
Chakraborty, S., Vafeiadis, V.: Grounding thin-air reads with event structures
  \textbf{3}(POPL) (2019). \doi{10.1145/3290383}

\bibitem{Chalupa2017}
Chalupa, M., Chatterjee, K., Pavlogiannis, A., Sinha, N., Vaidya, K.:
  Data-centric dynamic partial order reduction. Proc. ACM Program. Lang.
  \textbf{2}(POPL) (dec 2017). \doi{10.1145/3158119},
  \url{https://doi.org/10.1145/3158119}

\bibitem{DBLP:conf/popl/ClarkeES83}
Clarke, E.M., Emerson, E.A., Sistla, A.P.: Automatic verification of finite
  state concurrent systems using temporal logic specifications: {A} practical
  approach. In: Proceedings of POPL 1983. pp. 117--126. {ACM} Press (1983)

\bibitem{DBLP:journals/sttt/ClarkeGMP99}
Clarke, E.M., Grumberg, O., Minea, M., Peled, D.A.: State space reduction using
  partial order techniques. Int. J. Softw. Tools Technol. Transf.
  \textbf{2}(3),  279--287 (1999). \doi{10.1007/s100090050035},
  \url{https://doi.org/10.1007/s100090050035}

\bibitem{Faial}
Cogumbreiro, T., Lange, J., Rong, D.L.Z., Zicarelli, H.: Checking data-race
  freedom of gpu kernels, compositionally. In: Silva, A., Leino, K.R.M. (eds.)
  Computer Aided Verification. pp. 403--426. Springer International Publishing,
  Cham (2021)

\bibitem{Coutinho:2013}
Coutinho, B., Sampaio, D., Pereira, F.M., Meira~Jr, W.: Profiling divergences
  in gpu applications. Concurrency and Computation: Practice and Experience
  \textbf{25}(6),  775--789 (2013)

\bibitem{Eizenberg:2017}
Eizenberg, A., Peng, Y., Pigli, T., Mansky, W., Devietti, J.: Barracuda:
  binary-level analysis of runtime races in cuda programs. In: PLDI'17. PLDI
  2017, Association for Computing Machinery, New York, NY, USA (2017).
  \doi{10.1145/3062341.3062342}, \url{https://doi.org/10.1145/3062341.3062342}

\bibitem{DBLP:conf/popl/FlanaganG05}
Flanagan, C., Godefroid, P.: Dynamic partial-order reduction for model checking
  software. In: Palsberg, J., Abadi, M. (eds.) POPL'05. pp. 110--121. {ACM}
  (2005). \doi{10.1145/1040305.1040315},
  \url{https://doi.org/10.1145/1040305.1040315}

\bibitem{francis2014autonomous}
Francis, E.: Autonomous cars: no longer just science fiction (2014)

\bibitem{gaster2015hrf}
Gaster, B.R., Hower, D., Howes, L.: Hrf-relaxed: Adapting hrf to the
  complexities of industrial heterogeneous memory models. ACM Transactions on
  Architecture and Code Optimization (TACO)  \textbf{12}(1),  1--26 (2015)

\bibitem{Gibbons:1997}
Gibbons, P.B., Korach, E.: Testing shared memories. SIAM Journal on Computing
  \textbf{26}(4),  1208--1244 (1997). \doi{10.1137/S0097539794279614}

\bibitem{DBLP:books/sp/Godefroid96}
Godefroid, P.: Partial-Order Methods for the Verification of Concurrent Systems
  - An Approach to the State-Explosion Problem, Lecture Notes in Computer
  Science, vol.~1032. Springer (1996). \doi{10.1007/3-540-60761-7},
  \url{https://doi.org/10.1007/3-540-60761-7}

\bibitem{Goens:2023}
Goens, A., Chakraborty, S., Sarkar, S., Agarwal, S., Oswald, N., Nagarajan, V.:
  Compound memory models. Proc. ACM Program. Lang.  \textbf{7}(PLDI) (2023).
  \doi{10.1145/3591267}

\bibitem{haccrg}
Holey, A., Mekkat, V., Zhai, A.: Haccrg: Hardware-accelerated data race
  detection in gpus. In: 2013 42nd International Conference on Parallel
  Processing. pp. 60--69 (2013). \doi{10.1109/ICPP.2013.15}

\bibitem{hower2014heterogeneous}
Hower, D.R., Hechtman, B.A., Beckmann, B.M., Gaster, B.R., Hill, M.D.,
  Reinhardt, S.K., Wood, D.A.: Heterogeneous-race-free memory models. In:
  ASPLOS 2014. pp. 427--440 (2014)

\bibitem{cppstandard}
{ISO/IEC 14882}: Programming language {C++} (2011)

\bibitem{cstandard}
{ISO/IEC 9899}: Programming language {C} (2011)

\bibitem{Jeffrey:2016}
Jeffrey, A., Riely, J.: On thin air reads towards an event structures model of
  relaxed memory. In: LICS 2016. pp. 759--767 (2016).
  \doi{10.1145/2933575.2934536}

\bibitem{iguard}
Kamath, A.K., Basu, A.: iguard: In-gpu advanced race detection. In: SOSP'21. p.
  49–65 (2021). \doi{10.1145/3477132.3483545}

\bibitem{Kamath:2020}
Kamath, A.K., George, A.A., Basu, A.: Scord: A scoped race detector for gpus.
  In: 2020 ACM/IEEE 47th Annual International Symposium on Computer
  Architecture (ISCA). pp. 1036--1049 (2020).
  \doi{10.1109/ISCA45697.2020.00088}

\bibitem{Kang:2017}
Kang, J., Hur, C.K., Lahav, O., Vafeiadis, V., Dreyer, D.: A promising
  semantics for relaxed-memory concurrency. In: POPL'17. p. 175–189 (2017).
  \doi{10.1145/3009837.3009850}

\bibitem{Ketema:2017}
Ketema, J., Donaldson, A.F.: Termination analysis for {GPU} kernels. Sci.
  Comput. Program.  \textbf{148},  107--122 (2017).
  \doi{10.1016/J.SCICO.2017.04.009},
  \url{https://doi.org/10.1016/j.scico.2017.04.009}

\bibitem{DBLP:journals/pacmpl/Kokologiannakis18}
Kokologiannakis, M., Lahav, O., Sagonas, K., Vafeiadis, V.: Effective stateless
  model checking for {C/C++} concurrency. Proc. {ACM} Program. Lang.
  \textbf{2}({POPL}),  17:1--17:32 (2018). \doi{10.1145/3158105},
  \url{https://doi.org/10.1145/3158105}

\bibitem{DBLP:journals/pacmpl/Kokologiannakis22}
Kokologiannakis, M., Marmanis, I., Gladstein, V., Vafeiadis, V.: Truly
  stateless, optimal dynamic partial order reduction. Proc. {ACM} Program.
  Lang.  \textbf{6}({POPL}),  1--28 (2022). \doi{10.1145/3498711},
  \url{https://doi.org/10.1145/3498711}

\bibitem{DBLP:conf/pldi/Kokologiannakis19}
Kokologiannakis, M., Raad, A., Vafeiadis, V.: Model checking for weakly
  consistent libraries. In: Proceedings of PLDI 2019. pp. 96--110. {ACM} (2019)

\bibitem{kokologiannakis2021bam}
Kokologiannakis, M., Vafeiadis, V.: Bam: efficient model checking for barriers.
  In: International Conference on Networked Systems. pp. 223--239. Springer
  (2021)

\bibitem{Lahav:2015}
Lahav, O., Vafeiadis, V.: {Owicki-Gries} reasoning for weak memory models. In:
  ICALP'15. pp. 311--323 (2015). \doi{10.1007/978-3-662-47666-6\_25}

\bibitem{Lahav:2017}
Lahav, O., Vafeiadis, V., Kang, J., Hur, C.K., Dreyer, D.: Repairing sequential
  consistency in {C/C++11}. In: PLDI 2017. pp. 618--632 (2017).
  \doi{10.1145/3062341.3062352}, technical Appendix Available at
  \url{https://plv.mpi-sws.org/scfix/full.pdf}

\bibitem{lamport-sc}
Lamport, L.: How to make a multiprocessor computer that correctly executes
  multiprocess programs. {IEEE} Trans. Computers  \textbf{28}(9),  690--691
  (1979). \doi{10.1109/TC.1979.1675439}

\bibitem{Dartagnan2021}
Ponce-de Le{\'o}n, H., Haas, T., Meyer, R.: Dartagnan: Smt-based violation
  witness validation (competition contribution). In: Fisman, D., Rosu, G.
  (eds.) Tools and Algorithms for the Construction and Analysis of Systems. pp.
  418--423. Springer International Publishing, Cham (2022)

\bibitem{gpuharbor}
Levine, R., Cho, M., McKee, D., Quinn, A., Sorensen, T.: Gpuharbor: Testing gpu
  memory consistency at large (experience paper). In: ISSTA 2023. p. 779–791
  (2023). \doi{10.1145/3597926.3598095}

\bibitem{dat3m}
Ponce~de León, H.: Dat3m. \url{https://github.com/hernanponcedeleon/Dat3M}
  (2024)

\bibitem{pug1}
Li, G., Gopalakrishnan, G.: Scalable smt-based verification of gpu kernel
  functions. In: Proceedings of the Eighteenth ACM SIGSOFT International
  Symposium on Foundations of Software Engineering. p. 187–196. FSE '10,
  Association for Computing Machinery, New York, NY, USA (2010).
  \doi{10.1145/1882291.1882320}, \url{https://doi.org/10.1145/1882291.1882320}

\bibitem{pug2}
Li, G., Gopalakrishnan, G.: Parameterized verification of gpu kernel programs.
  In: 2012 IEEE 26th International Parallel and Distributed Processing
  Symposium Workshops \& PhD Forum. pp. 2450--2459 (2012).
  \doi{10.1109/IPDPSW.2012.302}

\bibitem{gklee}
Li, G., Li, P., Sawaya, G., Gopalakrishnan, G., Ghosh, I., Rajan, S.P.: Gklee:
  concolic verification and test generation for gpus. In: PPOPP 2012. p.
  215–224 (2012). \doi{10.1145/2145816.2145844}

\bibitem{sesa}
Li, P., Li, G., Gopalakrishnan, G.: Practical symbolic race checking of gpu
  programs. In: SC '14: Proceedings of the International Conference for High
  Performance Computing, Networking, Storage and Analysis. pp. 179--190 (2014).
  \doi{10.1109/SC.2014.20}

\bibitem{li2017ld}
Li, P., Hu, X., Chen, D., Brock, J., Luo, H., Zhang, E.Z., Ding, C.: Ld:
  low-overhead gpu race detection without access monitoring. ACM Transactions
  on Architecture and Code Optimization (TACO)  \textbf{14}(1),  1--25 (2017)

\bibitem{Lidbury:2017}
Lidbury, C., Donaldson, A.F.: Dynamic race detection for {C++11}. In: POPL
  2017. pp. 443--457 (2017). \doi{10.1145/3009837.3009857}

\bibitem{Luo:2021}
Luo, W., Demsky, B.: C11Tester: A Race Detector for C/C++ Atomics, p.
  630–646. Association for Computing Machinery, New York, NY, USA (2021).
  \doi{10.1145/3445814.3446711}

\bibitem{Lustig:2022}
Lustig, D., Cooksey, S., Giroux, O.: Mixed-proxy extensions for the {NVIDIA}
  {PTX} memory consistency model: industrial product. In: ISCA 2022 (2022).
  \doi{10.1145/3470496.3533045}

\bibitem{Lustig:2019}
Lustig, D., Sahasrabuddhe, S., Giroux, O.: A formal analysis of the {NVIDIA}
  {PTX} memory consistency model. In: ASPLOS 2019. pp. 257--270. {ACM} (2019).
  \doi{10.1145/3297858.3304043}, \url{https://doi.org/10.1145/3297858.3304043}

\bibitem{Mathur:lics2020}
Mathur, U., Pavlogiannis, A., Viswanathan, M.: The complexity of dynamic data
  race prediction. p. 713–727. LICS '20, Association for Computing Machinery,
  New York, NY, USA (2020). \doi{10.1145/3373718.3394783},
  \url{https://doi.org/10.1145/3373718.3394783}

\bibitem{10.1145/3563315}
Moiseenko, E., Kokologiannakis, M., Vafeiadis, V.: Model checking for a
  multi-execution memory model. Proc. ACM Program. Lang.  \textbf{6}(OOPSLA2)
  (Oct 2022). \doi{10.1145/3563315}, \url{https://doi.org/10.1145/3563315}

\bibitem{10.1145/2806886}
Norris, B., Demsky, B.: A practical approach for model checking c/c++11 code.
  ACM Trans. Program. Lang. Syst.  \textbf{38}(3) (May 2016).
  \doi{10.1145/2806886}, \url{https://doi.org/10.1145/2806886}

\bibitem{Orr:2015}
Orr, M.S., Che, S., Yilmazer, A., Beckmann, B.M., Hill, M.D., Wood, D.A.:
  Synchronization using remote-scope promotion. In: ASPLOS 2015. p. 73–86
  (2015). \doi{10.1145/2694344.2694350}

\bibitem{ozerk2022efficient_encrypt}
{\"O}zerk, {\"O}., Elgezen, C., Mert, A.C., {\"O}zt{\"u}rk, E., Sava{\c{s}},
  E.: Efficient number theoretic transform implementation on gpu for
  homomorphic encryption. The Journal of Supercomputing  \textbf{78}(2),
  2840--2872 (2022)

\bibitem{pandey2022transformational_drugdiscovery}
Pandey, M., Fernandez, M., Gentile, F., Isayev, O., Tropsha, A., Stern, A.C.,
  Cherkasov, A.: The transformational role of gpu computing and deep learning
  in drug discovery. Nature Machine Intelligence  \textbf{4}(3),  211--221
  (2022)

\bibitem{DBLP:conf/cav/Peled93}
Peled, D.A.: All from one, one for all: on model checking using
  representatives. In: Courcoubetis, C. (ed.) Computer Aided Verification, 5th
  International Conference, {CAV} '93, Elounda, Greece, June 28 - July 1, 1993,
  Proceedings. Lecture Notes in Computer Science, vol.~697, pp. 409--423.
  Springer (1993). \doi{10.1007/3-540-56922-7},
  \url{https://doi.org/10.1007/3-540-56922-7}

\bibitem{Peng:2018}
Peng, Y., Grover, V., Devietti, J.: Curd: a dynamic cuda race detector. In:
  PLDI 2018. p. 390–403. PLDI 2018, Association for Computing Machinery, New
  York, NY, USA (2018). \doi{10.1145/3192366.3192368},
  \url{https://doi.org/10.1145/3192366.3192368}

\bibitem{Podkopaev:2019}
Podkopaev, A., Lahav, O., Vafeiadis, V.: Bridging the gap between programming
  languages and hardware weak memory models. Proc. ACM Program. Lang.
  \textbf{3}(POPL) (2019). \doi{10.1145/3290382}

\bibitem{Pulte:2018}
Pulte, C., Flur, S., Deacon, W., French, J., Sarkar, S., Sewell, P.:
  Simplifying {ARM} concurrency: multicopy-atomic axiomatic and operational
  models for {ARMv8}. {PACMPL}  \textbf{2}({POPL}),  19:1--19:29 (2018).
  \doi{10.1145/3158107}

\bibitem{syclbook}
Reinders, J., Ashbaugh, B., Brodman, J., Kinsner, M., Pennycook, J., Tian, X.:
  Data Parallel C++: Programming Accelerated Systems Using C++ and SYCL.
  Apress, Berkeley, CA, 2 edn. (2023). \doi{10.1007/978-1-4842-9691-2},
  \url{https://doi.org/10.1007/978-1-4842-9691-2}, intel Corporation 2023

\bibitem{8916327_survey_ml}
Reuther, A., Michaleas, P., Jones, M., Gadepally, V., Samsi, S., Kepner, J.:
  Survey and benchmarking of machine learning accelerators. In: 2019 IEEE High
  Performance Extreme Computing Conference (HPEC). pp.~1--9 (2019).
  \doi{10.1109/HPEC.2019.8916327}

\bibitem{SorensenD16}
Sorensen, T., Donaldson, A.F.: Exposing errors related to weak memory in {GPU}
  applications. In: Krintz, C., Berger, E.D. (eds.) Proceedings of the 37th
  {ACM} {SIGPLAN} Conference on Programming Language Design and Implementation,
  {PLDI} 2016, Santa Barbara, CA, USA, June 13-17, 2016. pp. 100--113. {ACM}
  (2016). \doi{10.1145/2908080.2908114},
  \url{https://doi.org/10.1145/2908080.2908114}

\bibitem{Sorensen:oopsla:2016}
Sorensen, T., Donaldson, A.F., Batty, M., Gopalakrishnan, G., Rakamari\'{c},
  Z.: Portable inter-workgroup barrier synchronisation for gpus. p. 39–58.
  OOPSLA 2016, Association for Computing Machinery (2016).
  \doi{10.1145/2983990.2984032}

\bibitem{Sorensen:2018}
Sorensen, T., Evrard, H., Donaldson, A.F.: {GPU} schedulers: How fair is fair
  enough? In: Schewe, S., Zhang, L. (eds.) 29th International Conference on
  Concurrency Theory, {CONCUR} 2018, September 4-7, 2018, Beijing, China.
  LIPIcs, vol.~118, pp. 23:1--23:17 (2018). \doi{10.4230/LIPICS.CONCUR.2018.23}

\bibitem{Sorensen:2021}
Sorensen, T., Salvador, L.F., Raval, H., Evrard, H., Wickerson, J., Martonosi,
  M., Donaldson, A.F.: Specifying and testing {GPU} workgroup progress models.
  Proc. {ACM} Program. Lang.  \textbf{5}({OOPSLA}),  1--30 (2021).
  \doi{10.1145/3485508}, \url{https://doi.org/10.1145/3485508}

\bibitem{hernan24}
Tong, H., Gavrilenko, N., Ponce~de Leon, H., Heljanko, K.: Towards unified
  analysis of gpu consistency. In: Proceedings of the 29th ACM International
  Conference on Architectural Support for Programming Languages and Operating
  Systems, Volume 4. p. 329–344. ASPLOS '24, Association for Computing
  Machinery, New York, NY, USA (2025). \doi{10.1145/3622781.3674174},
  \url{https://doi.org/10.1145/3622781.3674174}

\bibitem{tunc_optimal_2023}
Tunç, H.C., Abdulla, P.A., Chakraborty, S., Krishna, S., Mathur, U.,
  Pavlogiannis, A.: Optimal {Reads}-{From} {Consistency} {Checking} for
  {C11}-{Style} {Memory} {Models}. Proceedings of the ACM on Programming
  Languages  \textbf{7}(PLDI),  137:761--137:785 (Jun 2023).
  \doi{10.1145/3591251}, \url{https://dl.acm.org/doi/10.1145/3591251}

\bibitem{simulee}
Wu, M., Ouyang, Y., Zhou, H., Zhang, L., Liu, C., Zhang, Y.: Simulee: detecting
  cuda synchronization bugs via memory-access modeling. In: ICSE'20. p.
  937–948 (2020). \doi{10.1145/3377811.3380358},
  \url{https://doi.org/10.1145/3377811.3380358}

\bibitem{xiao2010interxfbar}
Xiao, S., Feng, W.c.: Inter-block gpu communication via fast barrier
  synchronization. In: 2010 IEEE International Symposium on Parallel \&
  Distributed Processing (IPDPS). pp. 1--12. IEEE (2010)

\bibitem{arrayfire}
Yalamanchili, P., Arshad, U., Mohammed, Z., Garigipati, P., Entschev, P.,
  Kloppenborg, B., Malcolm, J., Melonakos, J.: {ArrayFire - A high performance
  software library for parallel computing with an easy-to-use API} (2015),
  \url{https://github.com/arrayfire/arrayfire}

\bibitem{10.1145/2737924.2737956}
Zhang, N., Kusano, M., Wang, C.: Dynamic partial order reduction for relaxed
  memory models. In: Proceedings of the 36th ACM SIGPLAN Conference on
  Programming Language Design and Implementation. p. 250–259. PLDI '15,
  Association for Computing Machinery, New York, NY, USA (2015).
  \doi{10.1145/2737924.2737956}, \url{https://doi.org/10.1145/2737924.2737956}

\bibitem{grace}
Zheng, M., Ravi, V.T., Qin, F., Agrawal, G.: Grace: a low-overhead mechanism
  for detecting data races in gpu programs. SIGPLAN Not.  \textbf{46}(8),
  135–146 (Feb 2011). \doi{10.1145/2038037.1941574},
  \url{https://doi.org/10.1145/2038037.1941574}

\bibitem{Zheng:2014}
Zheng, M., Ravi, V.T., Qin, F., Agrawal, G.: Gmrace: Detecting data races in
  gpu programs via a low-overhead scheme. IEEE Transactions on Parallel and
  Distributed Systems  \textbf{25}(1),  104--115 (2014).
  \doi{10.1109/TPDS.2013.44}

\end{thebibliography}

\newpage
\appendix
\onecolumn
\centerline{\Large{Appendix}}
This is the technical appendix of the submitted paper 
"$\projname$: A Stateless Model Checker for GPU Weak Memory Concurrency".

\textbf{Table of Content}
\begin{itemize}
\item \Cref{app:isoptimaldetails} Details on the $\isoptimal$ Procedure.
\item \Cref{app:complete} Completeness of the DPOR Algorithm. 
\begin{itemize}
\item \Cref{app:pred} Computing the unique predecessor of an execution graph. 
\item \Cref{app:prevproperties} Properties of the unique predecessor computation.
\item \Cref{app:completenessproofs} Contains the Lemmas and proofs for Completeness of the DPOR algorithm.     
\end{itemize}
\item \Cref{app:optmal} Proof for the Optimality of the DPOR Algorithm
\item \Cref{app:previllustration} Illustration on the Prev() procedure to compute the unique predecessor.
\item \Cref{app:expts} More details on the Experiments.

\end{itemize}




\clearpage

\subsection{More Details on the $\isoptimal$ Procedure}
\label{app:isoptimaldetails}

\[
        \label{prg:dpor-app}
        \inarrC{
        X=Y=Z=0;
        \\[1.2ex]
        \inparIII{
        \T_1\tup{\wg_1, \gpu_1}
        \\[1.2ex]
        r_0=X^\sys_\MOacq;
        \\[1.1ex]
        r_1=Y^\sys_\MOacq;
        }{
        ~~\T_2\tup{\wg_1, \gpu_1}
        \\[1.2ex]
        ~~Y^\sys_\MOrel=1;
        \\[1.1ex]
        }
        {
        ~~\T_3\tup{\wg_1, \gpu_1}
        \\[1.2ex]
        ~~Z^\sys_\MOrel=1;
        \\[1.1ex]
        ~~X^\sys_\MOrel=1;
        }
        }
\]

\begin{figure}[h]
\begin{center}
   \includegraphics[width=10cm]{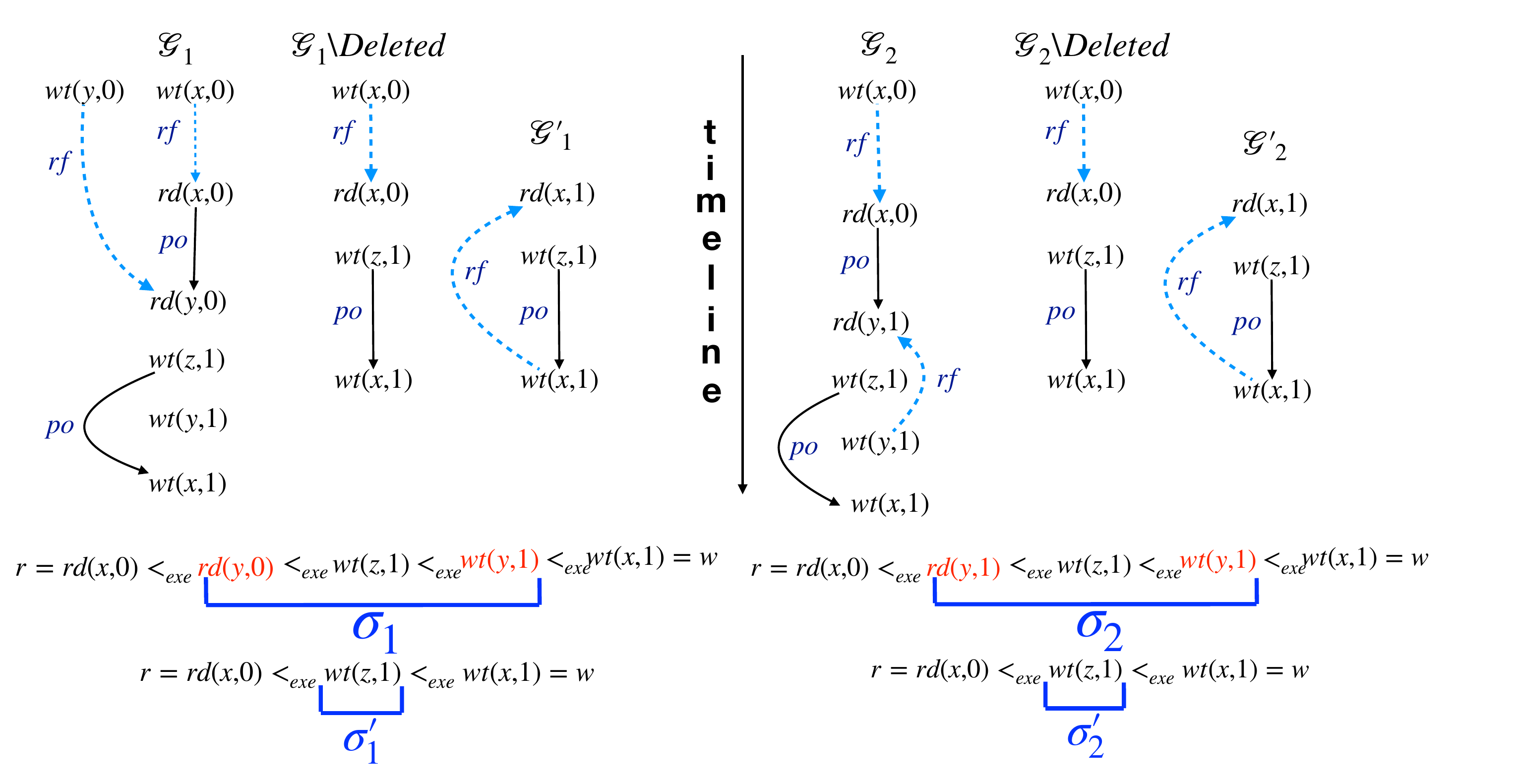}     
   \caption{Checks of $\isoptimal$ : \cref{eisread}. In the execution, we use the simpler notations $rd(x,v), wt(x,v)$ since the scope 
   and access is clear from the program.}
\label{isoptimal1-cond1}
\end{center}
\end{figure}

\noindent{\bf{The $\isoptimal$ Procedure}}. $\isoptimal$ is a procedure that ensures optimality without compromising the completeness of the algorithm : namely,  no execution graph is produced more than once, and each consistent execution graph is explored.  As part of the optimality check, 
we iterate through the events from $\deletedevents$ in the sequence $\sigma$ : recall that $\sigma$ has events $e$ such that $r \gorder e \gorder w$, where 
$w$ is the postponed write corresponding  to $r$. 
If $e \in \deletedevents$ is a read, then the write event $w_e$ it reads from should not be 
such that $e \gorder w_e$ and $w_e \in \deletedevents$ (\cref{eisread}) : this is done to avoid redundancy. 
The reasoning is that for each such graph $\G_1$ containing a read event $e \in \deletedevents$ reading from a
write $w_e$ such that $e \gorder w_e$, there is another execution graph $\G_2$ where $e$ reads from a write $w'_e$ such that (i) $w'_e \gorder e$ (such a $w'_e$ surely exists, for instance $w'_e$ could be the initial write), or (ii) $w'_e \in \po\rf.w$ (that is, $w'_e \notin \deletedevents$), and (iii) after removing 
events from $\deletedevents$ in $\G_1$ and $\G_2$ and adding the $\rf$ edge from $w$ to $r$, we obtain the same graph. 

Consider the  programs in  \Cref{isoptimal1-cond1}, \Cref{isoptimal1-cond2}. 
While Figure \ref{isoptimal1-cond1} illustrates (i), Figure \ref{isoptimal1-cond2} illustrates (ii). In both  graphs $\G_1, \G_2$ of Figures \ref{isoptimal1-cond1}, 
\ref{isoptimal1-cond2}, $w=\wt(x,1)$ is a postponed write for the read $r=\rd(x,0)$; they differ in the $\rf$ edge wrt $y$. The progress 
of time is from top to bottom as shown by the arrow.

\newpage 
\[
        \label{prg:dpor-app}
        \inarrC{
        X=Y=Z=0;
        \\[1.2ex]
        \inparIII{
        \T_1\tup{\wg_1, \gpu_1}
        \\[1.2ex]
        r_0=X^\sys_\MOacq;
        \\[1.1ex]
        r_1=Y^\sys_\MOacq;
        }{
        ~~\T_2\tup{\wg_1, \gpu_1}
        \\[1.2ex]
        ~~Y^\sys_\MOrel=1;
        \\[1.1ex]
        }
        {
        ~~\T_3\tup{\wg_1, \gpu_1}
        \\[1.2ex]
        ~~Y^\sys_\MOrel=2;
        \\[1.1ex]
        ~~X^\sys_\MOrel=1;
        }
        }
\]
\begin{figure}[h]
\begin{center}
   \includegraphics[width=10cm]{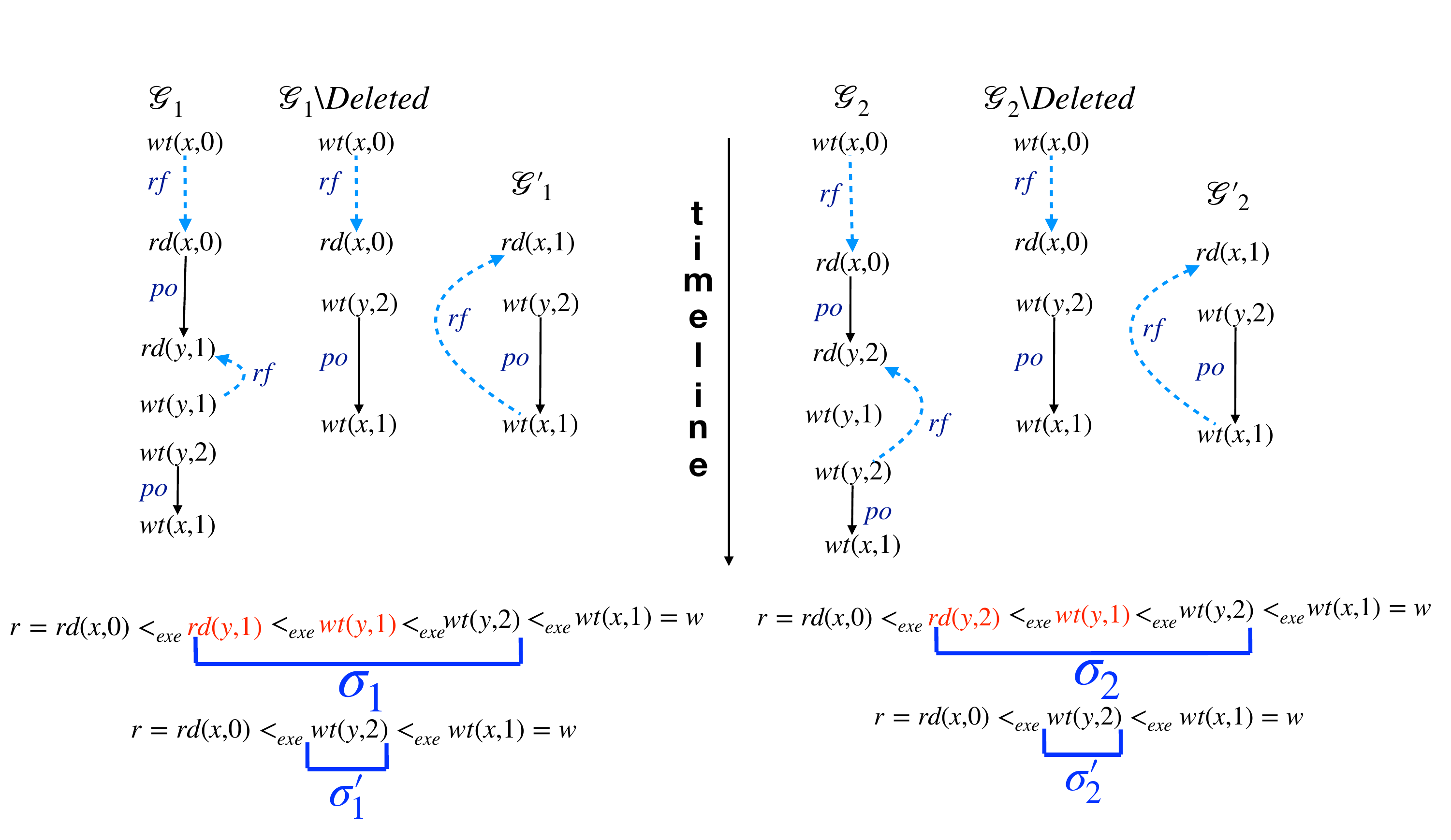}     \caption{Checks of $\isoptimal$ : \cref{eisread}. In the execution, we use the simpler notations $rd(x,v), wt(x,v)$ since the scope    and access is clear from the program. }
 \label{isoptimal1-cond2}
\end{center}
  \end{figure}

In Figure \ref{isoptimal1-cond1}, in $\G_1$, $y$ reads from the initial write $\wt(y,0)$ while in $\G_2$, it reads from the write 
 $\wt(y,1)$.  Considering the events between $r$ and $w$  in the $\gorder$ ordering, we have the sequences $\sigma_1$ 
and $\sigma_2$ respectively for $\G_1$ and $\G_2$
as shown in Figure \ref{isoptimal1-cond1}. The events shown in red in both sequences belong to $\deletedevents$ since they 
are not in $\po\rf.w$. On deleting these events, we obtain sequences $\sigma'_1$ and $\sigma_2'$ resulting in the same graph 
$\G'_1$ and $\G'_2$, resulting in redundancy.  Note that $\G_2$ violates the condition enforced in \cref{eisread} and 
$\isoptimal$ does not proceed further with it, instead it is enough to work with $\G_1$, illustrating (i). 
In Figure \ref{isoptimal1-cond2},
in $\G_1$, $y$ reads from $\wt(y,1) \in \deletedevents$, while 
in $\G_2$, $y$ reads from $\wt(y,2) \in \po\rf.w$. It can be seen that $\G_1$ violates the condition enforced in \cref{eisread}, and 
$\isoptimal$ does not proceed further with it, instead it is enough to work with $\G_2$, illustrating (ii).

\newpage 
\[
        \label{prg:dpor-app}
        \inarrC{
        X=Y=Z=0;
        \\[1.2ex]
        \inparIII{
        \T_1\tup{\wg_1, \gpu_1}
        \\[1.2ex]
        r_0=X^\sys_\MOacq;
        \\[1.1ex]
        Y^\sys_\MOrel=1;
        }{
        ~~\T_2\tup{\wg_1, \gpu_1}
        \\[1.2ex]
        ~~Y^\sys_\MOrel=2;
        \\[1.1ex]
        }
        {
        ~~\T_3\tup{\wg_1, \gpu_1}
        \\[1.2ex]
        ~~X^\sys_\MOrel=2;
        \\[1.1ex]
                }
        }
\]

\begin{figure*}[h]
\begin{center}
   \includegraphics[width=10cm]{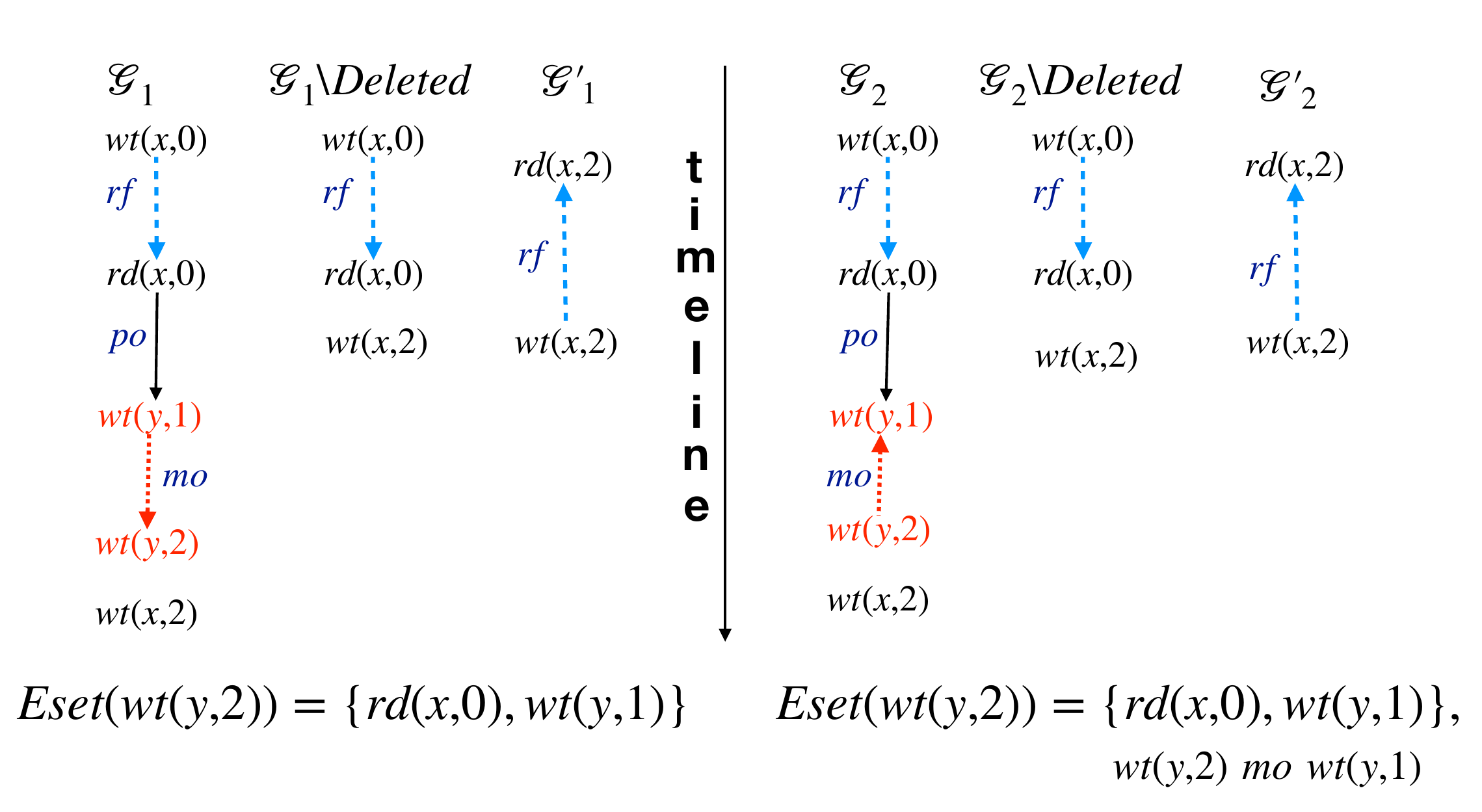}     \caption{Checks of $\isoptimal$ : \cref{line6check}. In the execution, we use the simpler notations $rd(x,v), wt(x,v)$ since the scope 
   and access is clear from the program.}
 \label{isoptimal-Eset-cond1}
\end{center}
  \end{figure*}

\noindent{\bf{The $\co$-maximality check for $\deletedevents$}}.  The second optimality check pertains to the orientation of $\co$ edges of $\deletedevents$.
 For optimality, we need only consider those execution graphs $\G$ where all writes $\wt(x,v) \in \deletedevents$ are $\co_x$-maximal (that is, there is no write $\wt(x,v')$ in $\G$ such that $\wt(x,v)~\co~\wt(x,v')$) and 
 for all $\rd(x,v) \in \deletedevents$, the corresponding $\wt(x,v)$ they read from, are also $\co_x$-maximal.  The intuition is that a write in $\deletedevents$ will be removed and added later during exploration; at that time, being a later write, this will be $\co$-maximal wrt the earlier and undeleted writes. Likewise, 
 a read event in $\deletedevents$ should read from the $\co$-maximal write : since it will be removed, when it is added later in the exploration, 
 it has to read from the maximal write.

 Given an execution graph $\G$ containing a  read $r$, and its postponed write $w$, 
 for  each write event $e \in \deletedevents$, $Eset(e)$ is the set of all events 
  $e'$ which were visited earlier than $e$ in $\G$  ($e' \gorder e$) or those events which are in $\po\rf.w$.
Line \ref{line6check} checks whether each write $e=\wt(x,v)$ in $\deletedevents$ is $\co_x$-maximal wrt $Eset(e)$, that 
is, there is no $e' \in Eset(e)$ such that $e~ \co_x ~e'$. Likewise for each read $e=\rd(x,v) \in \deletedevents$, having $e''=\wt(x,v)$
as its corresponding write, \cref{line6check} checks whether there is some $e' \in Eset(e)$ such that $e'' ~\co_x~ e'$. If yes, the algorithm does not continue exploration with $\G$. 
Figures \ref{isoptimal-Eset-cond1}, \ref{isoptimal-Eset-cond2} illustrate the case with two graphs $\G_1, \G_2$, where   $r=\rd(x,0)$, and $w=\wt(x,2)$ is its postponed read. In Figure \ref{isoptimal-Eset-cond1},
$\G_2$ violates \cref{line6check} : $e=\wt(y,2) \in \deletedevents$ is not $\co_y$-maximal wrt $Eset(e)$ since 
$e ~\co_y ~\wt(y,1)$ while $\wt(y,1) \gorder e$. Indeed, it suffices to continue exploration with $\G_1$, as both graphs $\G_1, \G_2$ yield the same graph 
after removing the events from $\deletedevents$ and assigning $w ~\rf ~r$. Likewise, 
in Figure \ref{isoptimal-Eset-cond2}, $\G_1$ violates \cref{line6check} : $e=\wt(y,1) \in \deletedevents$ is not $\co_y$-maximal wrt $Eset(e)$ since 
$e ~\co_y ~\wt(y,2)$ while $\wt(y,2) \in \po\rf.w$. It suffices to continue exploration with $\G_2$, as both graphs $\G_1, \G_2$ yield the same graph 
after removing the events from $\deletedevents$ and assigning $w ~\rf ~r$.
 \[
        \label{prg:dpor-app}
        \inarrC{
        X=Y=Z=0;
        \\[1.2ex]
        \inparII{
        \T_1\tup{\wg_1, \gpu_1}
        \\[1.2ex]
        r_0=X^\wg_\MOrlx;
        \\[1.1ex]
        Y^\wg_\MOrlx=1;
        }{
        ~~\T_2\tup{\wg_1, \gpu_1}
        \\[1.2ex]
        ~~Y^\wg_\MOrlx=2;
        \\[1.1ex]
        ~~X^\wg_\MOrlx=2;
        }
        }
\]
\begin{figure}[h]
\begin{center}
   \includegraphics[width=11cm]{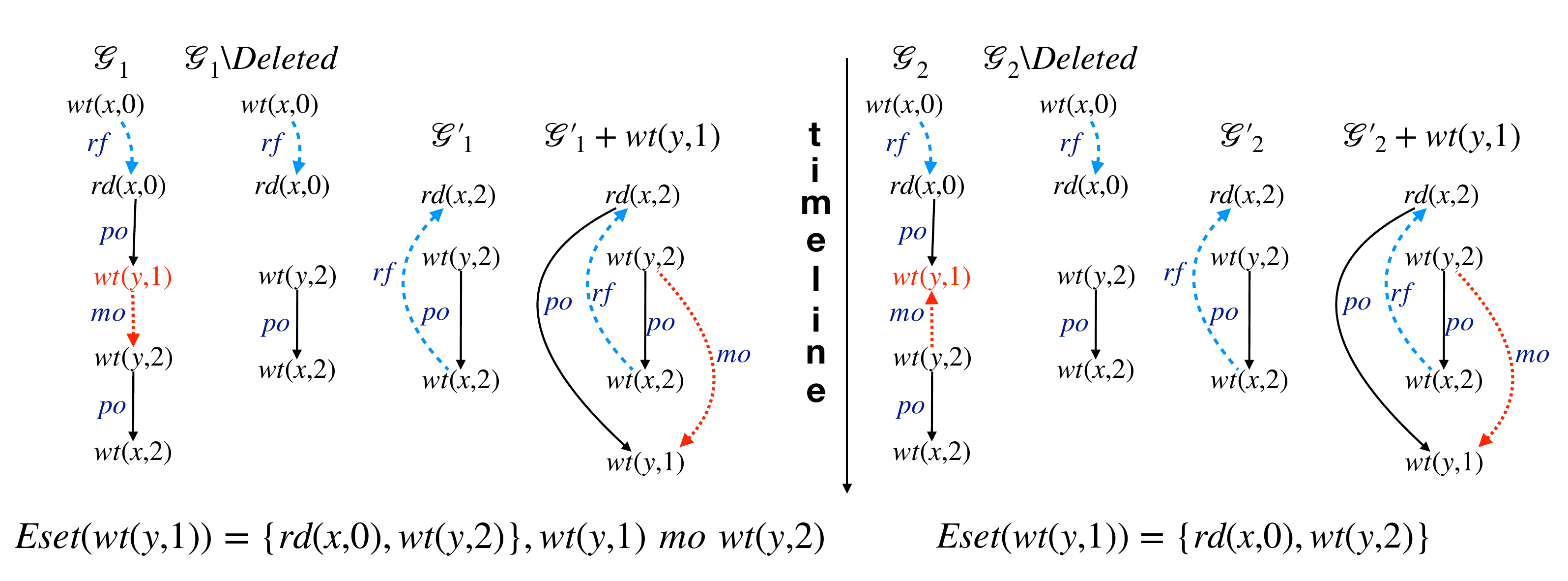}     \caption{Checks of $\isoptimal$ : \cref{line6check}. In the execution, we use the simpler notations $rd(x,v), wt(x,v)$ since the scope 
   and access is clear from the program.}
   \label{isoptimal-Eset-cond2}
   \end{center}
 \end{figure}

\subsection{Completeness of the DPOR Algorithm}
\label{app:complete}
Since the execution is represented as the graph, throughout the appendix, we will use the terms "execution graph"  and "execution" interchangeably to for the execution $\egraph=\tup{\E,\po,\rf,\co,\rmw}$.

Throughout this appendix, let ($\egraph, \le_{\egraph}$)  denote 
a total order extension of  $\egraph$.

Given an an execution graph $\egraph$, we define set enabled($\egraph$) to be the set of events $(n,tid,-,-,-,-,-)$ such that thread
$tid$ contains $n-1$ events. 
That is, enabled($\egraph$) is the set of events available to be executed from each thread.

We assume $\nextevent()$ returns the next event according to some total order called as next order and denoted as $\nev$ this.

Given an execution graph $\egraph$, we define the next event of $\egraph$, $\nextevent()$, 
as the minimum wrt $\nev$ from the set enabled($\egraph$).

For read event $e$ and write event $w$ such that $e ~\egraph~w $, $\reversed(\egraph,e,w)$ is true iff $e$ reads from $w$.
For read event $\reversed(\egraph,e,-)$ is true if read event $e$ reads from some postponed write $w \in \egraph$
Similarly, $\reversed(\egraph,-,w)$ is true if there exists a read event $e ~\gorder~ w $ and $e$ reads from $w$.
Similarly, we define $\reversed(\egraph,-,-)$ for some read $e$ and write $w$ such that $e~\gorder~w$ and $w~\rf~e$.

We define an execution graph $\egraph$ as full if $\nextevent(\egraph) = \bot$ and $\tt{Consistent}(\egraph)$.

We define a write event $w$ on location $x$ as maximal, denoted as $max_{\co_{x}}$, iff $\nexists w' (w ~\co_x~ w')$.

We use $\egraph|_{E'}$ to represent execution graph $\egraph$ restricted to the events in $E' \subseteq E$.


We prove completeness using the following steps. 
\begin{enumerate}
    \item We define a function $Prev()$ over execution graphs. 
    Given a consistent execution graph, 
    $\egraph$, $Prev(\egraph)$ returns the previous execution graph $\egraph'$ 
    that led to the recursive call $\explore(\egraph)$. 
    Then we prove the soundness of $Prev()$. 
    \item We prove that there exists a finite sequence of execution graphs 
    H=$\{\egraph_0, 
    \egraph_1, \ldots, \egraph_k\}$ such that $\egraph_i = Prev(\egraph_{i+1})$, 
    $\egraph_0 = \egraph_{\emptyset}$ and $\egraph_k = \egraph$.
\end{enumerate}

\subsubsection{Computing the unique predecessor of an execution graph} 
\label{app:pred}

In this section, we define the function  Prev, which computes the unique predecessor of the $\egraph$. 

First, we define $\lastofg{\egraph}$ as  the last event of the $\egraph$ as follows:
\begin{itemize}
    \item If $\egraph = \emptyset$ then $\lastofg{\egraph} = \emptyset$.
    \item $\lastofg{\egraph} = e$, where $e = max_{\nev}\{e' | e' $ is $\po\rf$ maximal in $\egraph \}$. 
    Event $e$ is $\po\rf$ maximal, if $e$ has no successor in $\po\rf$. 
\end{itemize}

\SetAlCapFnt{}
\begin{algorithm}[ht]
{
\DontPrintSemicolon
\SetKwFunction{FProc}{Prev}
\SetKwProg{Proc}{Procedure}{}{}
\Proc{\FProc{$\pgm$, $\egraph$}}{
    \If{$\egraph == \emptyset$}{
        \textbf{return} $\emptyset$
    }
    $e \leftarrow \lastofg{\egraph}$

    \If(\tcp*[f]{\scriptsize{$w=RF(e)$}}){ $e \in \rd \wedge \reversed(\egraph,e,w)$}{
        
        $\egraph' = \egraph \backslash \{e,w\})$ 
        
        \textbf{return} MaxEGraph($\pgm$,$\egraph' $, $w$)
    }
    \Else{
         \textbf{return} $\egraph \, \backslash \, \{e\}$
    }
}
\SetKwFunction{FProc}{MaxEGraph}
\SetKwProg{Proc}{Procedure}{}{}
\Proc{\FProc{$\pgm$, $\egraph$ , $w$}}{
    $a$ $\leftarrow$ \texttt{NextEvent} ( $\pgm$, $\egraph$) 

    \If{$a==w)$}{
        \textbf{return} $\egraph$
    }

    \Switch{$a$}{
    
    \Case{$a = \wt^\sco_o(X,v)$} {
        $\egraph' = \tt{addMO}(\egraph,a)$

         $ \tt{SetAsMOLatestWrite(\egraph' , x,a)} $
         
        \textbf{return} MaxEGraph$(\pgm , \egraph' , w)$
    }

    \Case{$a=\rd^\sco_o(X,v)$} {
       
            \textbf{let} $w$ be s.t. \readfromlatest$(\egraph,a,w)$ holds
            
            \textbf{return} MaxEGraph$(\pgm ,\addrf(\egraph,w,a) ,w)$
            
    }

    \Other (\tcp*[f]{$a = \_$})  {
         \textbf{return} MaxEGraph$(\pgm , \egraph ,w)$
    }
}
}

}
\caption{Prev($\pgm$, $\egraph$ ) \label{alg:prev}}
\end{algorithm}

Algorithm \ref{alg:prev} gives the procedure to get the predecessor of the execution graph $\egraph$. 
We use Prev($\egraph$) to denote  algorithm \ref{alg:prev} which takes parameters Prev($\pgm,\egraph)$. 
First, we define two functions.
\begin{enumerate}
\item \readfromlatest$(\egraph,a,w)$ evaluates to true iff $w = max_{\co_{\varof{a}}}$.
$ \addrf$ is as in $\explore$.
\item $ \tt{SetAsMOLatestWrite(\egraph , x,w)} $ updates $max_{\co_{x}} =w$.
\end{enumerate}

For execution graphs $\egraph=\tup{\E,\po,\rf,\co,\rmw}$ and $\egraph'=\tup{\E',\po',\rf',\co',\rmw'}$
we write $\egraph \subseteq \egraph$ to denote 
$\E  \subseteq \E'$, $\po \subseteq \po'$, 
$\rf \subseteq \rf'$ , $\co \subseteq \co'$ 
and $\rmw \subseteq \rmw'$.

Similarly, for execution graphs $\egraph=\tup{\E,\po,\rf,\co,\rmw}$ and $\egraph'=\tup{\E',\po',\rf',\co',\rmw'}$
we write $\egraph \equiv \egraph$ to denote 
$\E  = \E'$, $\po = \po'$, 
$\rf = \rf'$ , $\co = \co'$ 
and $\rmw = \rmw'$.

We use $\egraph \prevstep \egraph'$ to denote that $\egraph =$ Prev($\egraph'$). 
Similarly, we use $\egraph \prevsteps \egraph'$ to $\egraph$ is obtained after a sequence of $\prevstep$ from $\egraph'$.

Now we show that the Prev() algorithm is sound: if the $\egraph$ is consistent, then the Prev($\egraph$) is consistent.
Then we argue that if Prev($\egraph$) is explored by the DPOR algorithm, then $\egraph$ will also be explored.

\subsubsection{Properties of Prev() computation}
\label{app:prevproperties}

\begin{lemma}
    For consistent and full $\egraph_t$, if $\egraph \prevsteps \egraph_t$ then $\egraph$ is consistent.
    \label{lem:consprev}
\end{lemma}

\begin{proof}
We prove this using induction on the length of the sequence $\egraph \prevsteps \egraph_t$.

\noindent  \textbf{Base Case:}  For $\egraph$ = $\egraph_t$, it holds trivially.
    
\noindent \textbf{Inductive Hypothesis:} Assume $\egraph \prevstep \egraph'$ and $\egraph'$ is consistent.
 Let $e=\lastofg{\egraph'}$.
We consider two cases:
\begin{itemize}
    \item $\neg(\reversed(\egraph',e,-))$: Since  $\neg(\reversed(\egraph',e,-))$, $\egraph = \egraph' \backslash \{e\}$.
    Let $\egraph=\tup{\E,\po,\rf,\co,\rmw}$ and $\egraph'=\tup{\E',\po',\rf',\co',\rmw'}$.
    Since $\egraph = \egraph \backslash \{e\}$, we have 
    $\E=\E'\backslash\{e\}$, $\po \subseteq \po'$, $\rf \subseteq \rf'$ , $\co \subseteq \co'$ and $\rmw = \rmw'$.
    That is, $\egraph \subseteq \egraph'$.
    Since $\egraph'$ is consistent, we have the following:
    \begin{itemize}
        \item $\lhb';\eco'^?$ is irreflexive .
        \item $\rmw'\cap (\fr';\co')$ is empty.
        \item $(\incl \cap \psc')$ is acyclic.
        \item $\po' \cup \rf'$ is acyclic.
    \end{itemize}
    
    Since $\egraph \subseteq \egraph'$, it follows that we have 
     \begin{itemize}
        \item $\lhb;\eco^?$ is irreflexive .
        \item $\rmw\cap (\fr;\co)$ is empty.
        \item $(\incl \cap \psc)$ is acyclic.
        \item $\po \cup \rf$ is acyclic.
    \end{itemize}
    Hence, $\egraph_p$ is consistent.

    \item $\reversed(\egraph',e,w)==True$ : 
    Let $\egraph'' = \egraph' \backslash \{e,w\}$. 
    It follows that $\egraph''$ is consistent (similar to the first case).
    All read (write) events added during the MaxEgraph($\egraph''$) reads from (are $\co$-after) the maximal write on the corresponding location.
    Hence, after each step, the corresponding is consistent.
    Therefore, $\egraph$ is consistent.
\end{itemize}

We conclude from all the above cases that $\egraph_p$ is consistent.
\end{proof}

In the following lemma, we show that $\lastofg{\egraph}$ coincides with $\nextevent($Prev($\egraph$)$)$.

\begin{lemma}
    For consistent and full $\egraph_t$, if $\egraph \prevstep \egraph' \prevsteps \egraph_t$, and $e=\lastofg{\egraph'}$ then $\nextevent(\egraph)=e$ if $\neg(\reversed(\egraph',e,-))$, otherwise $\nextevent(\egraph)=w$ where $\reversed(\egraph',e,w)==True$ .
    \label{lem:nextprev-eq-lastg}
\end{lemma}

\begin{proof}
     Let $e=\lastofg{\egraph'}$.
We consider two cases:
\begin{itemize}
\item $\neg(\reversed(\egraph',e,-))$: We have $\egraph = \egraph' \backslash \{e\}$. 
We have $e = \lastofg{\egraph'}$.
It follows that $\nexists e' \in $ enabled($\egraph'$) such that $e'~\nev~e$.
($\nextevent()$ returns minimum event wrt $\nev$) from the set enabled()).
Since, $\egraph = \egraph' \backslash \{e\}$ we have enabled($\egraph$) $\\$ enabled($\egraph'$) $= e$,
Hence, it follows that $\nexists e' \in $ enabled($\egraph$) such that $e'~\nev~e$.
Hence, $\nextevent(\egraph) = e$.

\item $\reversed(\egraph',e,w)$: Follows from the construction of Prev($\egraph'$). (lines 12-13 of Prev() Algorithm \ref{alg:prev})

\end{itemize}

\end{proof}

In the following lemma, we show that if Prev($\egraph$) is explored by the DPOR algorithm, then $\egraph$ will also be explored.
    
\begin{lemma}
    For every consistent $\egraph$,  if $\egraph'_p \prevstep \egraph$ and 
    $\egraph_p$ is reachable by the DPOR algorithm,  then $\egraph$ is also reachable.
    \label{lem:reachblity}
\end{lemma}

\begin{proof}
    Let $e=\lastofg{\egraph}$.
    We consider the following cases. 
    \begin{itemize}
        \item $\neg(\reversed(\egraph',e,-))$:
            Since $\neg(\reversed(\egraph',e,-))$, we have $\egraph_p = \egraph \backslash \{e\}$.
            From Lemma \ref{lem:nextprev-eq-lastg}, we have $\nextevent(\egraph_p)=e$.
            Hence, the algorithm explores $\egraph$ by invoking the $\explore(\pgm,\egraph)$ from $\explore(\pgm,\egraph_p)$ (Lines \ref{exploreafter-co} and \ref{exploreafter-rf} of Algorithm \ref{alg:algexplore})
        \item $(\reversed(\egraph',e,w))$:
            Since $(\reversed(\egraph',e,w))$ we have $w=\nextevent(\egraph)$. 
            Hence, the DPOR algorithm will invoke $\declarepostponed(\pgm,\egraph_p,w)$ (Line \ref{postponedrfs} of Algorithm \ref{alg:algexplore}) to explore the reads from which $w$ is delayed, including read event $e$.
            Let $\egraph_n$ be the consistent execution graph explored during $\delayedrf(\pgm,\egraph_p,w)$ 
            such that the corresponding delayed read is $e$ with $w~\rf~e$.
            Let $\deletedevents_n$ be the set of deleted events from $\egraph$ to generate $\egraph_n$.
            Let $\deletedevents_p$ be the set of added events during the construction of $\egraph_p =$ Prev($\egraph$).
            Now, we need to show that $\egraph_n \equiv \egraph$.
            Let $\egraph=\tup{\E,\po,\rf,\co,\rmw}$, $\egraph_n=\tup{\E_n,\po_n,\rf_n,\co_n,\rmw_n}$ 
            and $\egraph_p=\tup{\E_p,\po_p,\rf_p,\co_p,\rmw_p}$.
            From the construction of $\egraph_p$ from $\egraph$ and $\egraph_n$ from $\egraph_p$,
            we have $\egraph|_{E_p\cap E} \equiv \egraph_n|_{E_p\cap E}$.
            We need to show that $E \backslash E_n = E \backslash E_p$. 
            That is, we need to show that $\deletedevents_p = \deletedevents_n$.
    
            Assume $ev \in \deletedevents_n$.
            Since $ev \in \deletedevents_n$  it follows that 
            $e \gorder_p ev \, \wedge \, ev \gorder_p w \wedge ev \notin \po_p\rf_p.w$
            Assume $ev \notin \deletedevents_p$. 
            Since, $ev \notin \deletedevents_p$ we consider the following cases:
            \begin{itemize}
                \item $e \in ~\po_p\rf_p.ev$: Since $e \po_p\rf_p ev$ and $ev \in E_p$ and $ev \notin \deletedevents_p$, we have
                $ev \in E$. 
                Since $e=\lastofg{\egraph}$, there is no $\po\rf$ successor for $e$ in $\egraph$.
                This leads to the contradiction.
                \item $ev~\rf_p~e$: Since $e \gorder_p ev$ , $ev~\rf_p~e$, if follows $\isoptimal(\egraph_p,\deletedevents_n\cup\{e\},w,e)$ will return false(Line \ref{eisread} of Algorithm \ref{alg:algisopt}).
                This leads to the contradiction.
                \item $ev \in \po_p\rf_p.e$:  
                Since $ev \po_p\rf_p e$ and $e \gorder_p ev$, there exists write event $ew$ and read event $er ~\gorder_p~e$ such that 
                $ev \po_p\rf_p ew$, $er \po_p\rf_p e$ and $ew ~\rf_p~er$. This implies that $\isoptimal(\egraph_p,\deletedevents_n\cup\{e\},w,e)$ will return false(Line \ref{line6check} of Algorithm \ref{alg:algisopt}).
    
                \item $\neg(e \po_p\rf_p.ev)$ , $\neg(ev \po_p\rf_p e)$ and $ev \nev e$: $ev \nev e$ This implies that $\ev \gorder e$ which leads to contradiction.
                \item $\neg(e \po_p\rf_p.ev)$ , $\neg(ev \po_p\rf_p e)$ and $e \nev ev$: Since $e = \lastofg{\egraph}$, it follows that $\nexists ev' \in \egraph$ 
                 such that $e \nev ev'.$. This leads to a contradiction since $ev \in \egraph$.
            \end{itemize}
        
            Assume $ev \in \deletedevents_p$.
            Since $ev \in \deletedevents_p$, from the construction of $\egraph_p$ from $\egraph$, we have $e \nev ev$. 
            Hence, it follows that $e \gorder_p ev$.
            Since $ev \in \deletedevents_p$, we have $ev \notin \E$, it is not possible that $ev \notin \po_p\rf_p.w$.
            Since $e \gorder_p ev$ and $ev \notin \po_p\rf_p.w$ it follows that $ev \in \deletedevents_n$ for $\declarepostponed(\pgm,\egraph_p,w)$ and read $e$ (Line \ref{erased} of Algorithm \ref{alg:algprfs}).
        
    \end{itemize}

\end{proof}

\smallskip

From Lemmas \ref{lem:consprev} and \ref{lem:reachblity} it follows that 
for $r=\lastofg{\egraph}$ if $\reversed(r)=true$  in $\egraph$ 
then $\explore(\pgm,\egraph_p)$ calls $\declarepostponed(\pgm,\egraph_p,RF(e))$ which in turn invokes $\explore(\pgm,\egraph')$ such that $\egraph'\equiv \egraph$.

For each step performed by Algorithm \ref{alg:algexplore},  we keep track of the number of
$\declarepostponed$ performed in a set we call as $\pprfs$. 
So whenever Algorithm\ref{alg:algexplore}
performs $\declarepostponed$ of read $r$ from the write $w$, we add $(r,w)$ to $\pprfs$.
Similarly, when we get the predecessor Prev$(\egraph)$ of $\egraph$ and $\lastofg{\egraph} = r$ and $\reversed(r)=true$ with $w=RF(r)$, 
we remove $(r,w)$ from the set $\pprfs$.
We use $\pprfs_\egraph$ to denote the $\pprfs$ 
till the $\egraph$.

\noindent For an execution graph $\egraph$, a write $w$ can be postponed for a unique read $r$. 
Hence there exists no $r' \neq r$ such that $(r',w) \in \pprfs_{\egraph}$ if $(r,w) \in \pprfs_{\egraph}$.

\subsubsection{Completeness of the DPRO algorithm}
\label{app:completenessproofs}

\begin{lemma}
    Consider a sequence $(\egraph_0,\pprfs_0), (\egraph_1,\pprfs_1), \cdots (\egraph_n,\pprfs_n)$ 
    such that $\egraph_{i-1}=$Prev($\egraph_i$) 
    for every $i:0 < i \le n$ and $\egraph_n$ is full and consistent. 
    If $\reversed(ev_1)=true$ in $\egraph_1$ for $ev_1 = \lastofg{\egraph_1}$
  then    $\pprfs_{\egraph_0} \supset$  $\pprfs_{\egraph_1}$, otherwise $\egraph_0 = \egraph_1 \backslash \{ev_1\}$.
    \label{lem:decrevents}
\end{lemma}

\begin{proof}
    Let $ev_1 = \lastofg{\egraph_1}$.
    If $\neg(\reversed(ev_1))$ in $\egraph_1$, then $\egraph_0$ = Prev($\egraph_1$) = $\egraph_1\backslash\{ev_1\}$. (lines 9 of Algorithm \ref{alg:prev})
        Assume otherwise, that is $\reversed(ev_1)$ in $\egraph_1$. 
    Assume we have $\delayedrf(\egraph_1,r=ev_1)$ and $w=RF(r)$.
    
   Assume $\pprfs_{\egraph_0} \subseteq$  $\pprfs_{\egraph_1}$.
   From the construction of Prev($\egraph_1$) we have $\pprfs_{\egraph_0}$ $ \supseteq$  $\pprfs_{\egraph_1}$. 
   Hence it follows that $\pprfs_{\egraph_0} =$  $\pprfs_{\egraph_1}$.
   Since we have $\delayedrf$ $(\egraph_1,r)$ and $\pprfs_{\egraph_0}$ =   $\pprfs_{\egraph}$, 
   it follows that $(r,w) \in \pprfs_{\egraph_0}$ and $(r,w) \in \pprfs_{\egraph_1}$.
   Let $(\egraph_i,\pprfs_i)$ be such that $w \notin \egraph_i$ or $(r,w) \notin \pprfs_i$. 
   Hence, there exists $(\egraph_j,\pprfs_j)$ 
   such that $w \in \egraph_j$ and $(r,w) \in \pprfs_j$ 
   and $(\egraph_0,\pprfs_0), (\egraph_1,\pprfs_1),$ 
   $ \cdots (\egraph_j,\pprfs_j)$ $,(\egraph_i,\pprfs_i)$ 
   $\cdots (\egraph_n,\pprfs_n)$.
   
   \noindent Then we have $(r,w) \in \pprfs_i$ and $w \notin \egraph_i$
   Hence, it follows that $\egraph_j \equiv $Prev($\egraph_i$) and assume $\delayedrf(\egraph_i, r')$.

   We now prove that for all $(\egraph_k,\pprfs_k)$ in the sequence $(\egraph_1,\pprfs_1)$ $\cdots$ $(\egraph_k,\pprfs_k)$ $\cdots$ $(\egraph_j,\pprfs_j)$, there is no read $ev_r$ such that $w~\rf ~ev_r$ and $ev_r \gorder w$.

    \begin{itemize}
        \item[$\bullet$] $(\egraph_k,\pprfs_k)$ = $\egraph_j,\pprfs_j$:
        Since $\egraph_j$ = Prev($\egraph_i$), $w \notin \egraph_i$
        and $w \in \egraph$.
        That is, $w$ is explored during computation of Prev($\egraph_i$).
        All events added during Prev() 
        are $\po\rf$-maximal. 
        Hence for any read $r'$ and write $w'$ added during Prev(), such that $r' \gorder w'$, it is not possible that $w'~ \rf ~r'$. 
        \item[$\bullet$] Consider the sequence $(\egraph_1,\pprfs_1)$ $\cdots$ $(\egraph_{k_1},\pprfs_{k_1})$ , $(\egraph_{k_2},\pprfs_{k_2})$ $\cdots$ $(\egraph_j,\pprfs_j)$ such that  the property holds for $(\egraph_{k_2},\pprfs_{k_2})$. 
        We prove it holds for $(\egraph_{k_1},\pprfs_{k_1})$.
        We have $(\egraph_{k_1}$ $,\pprfs_{k_1})$ = Prev($\egraph_2,\pprfs_2$).
        Let $ev_{k_2} = \lastofg{\egraph_{k_2}}$.
        If $\neg(\reversed(ev_{k_2}))$, then $\egraph_{k_1}$ = $\egraph_{k_2} \backslash ev_{k_2}$.
        Hence, the property holds.
        If $(\reversed(ev_{k_2})$, then all the events added to the $\egraph_{k_1}$ are $\gorder$-after $w$. Hence, the property holds.
    \end{itemize}
    Hence, for $\egraph_1$ there exists no $ev_r$ such that $w ~\rf~ ev_r$ and $ev_r ~\gorder~ w$.
    This leads to the contradiction, because $(w ~\rf~ r) \in \egraph_1$.
    
   Therefore  $\pprfs_{\egraph_0} \supset$  $\pprfs_{\egraph_1}$.
\end{proof}

\begin{lemma}
    For every consistent and full $\egraph_t$, we have $\egraph_{\emptyset} \prevsteps \egraph_t$.
    \label{lem:finiteprev}
\end{lemma}

\begin{proof}
    The proof follows from lemma \ref{lem:decrevents}.
    Consider $\egraph_p$ = Prev($\egraph$). Let $ev$ = $\lastofg{\egraph}$.
    If $\neg(\reversed(ev))$, then $\egraph_p = \egraph \backslash ev$.
    If $\reversed(ev)$, then by lemma \ref{lem:decrevents}, $|\pprfs_\egraph| < |\pprfs_{\egraph_p}|$. 
    Hence, for every Prev($\egraph$) ($|\pprfs_\egraph|,-|\egraph_{\eventset}|$) strictly increases and is bounded(each execution of the program is of finite length).
\end{proof}

For an execution graph $\egraph$, we use Prev$^k$($\egraph$) to denote the execution graph $ \egraph_k$ obtained after applying Prev() $k$ times starting from the $\egraph$. That is, Prev$^k$($\egraph$) = $\egraph_k$ where we have the sequence $\egraph_k \prevstep \egraph_{k-1} \cdots \prevstep \egraph_1 \prevstep \egraph$

\begin{lemma}
    For every consistent execution graph $\egraph$, the algorithm will 
    generate a sequence of execution graphs $\egraph_0,\egraph_1\cdots\egraph_n$,
    where $n \in \mathcal{N}$,$\egraph_0 = \emptyset$ and $\egraph_n$ = $\egraph$.
    \label{lem:prevseq}
\end{lemma}

\begin{proof}
    Let $\tt{P}$ = Prev$^n$($\egraph$) $,$ Prev$^{n-1}$($\egraph$) $,\cdots$
    Prev$^0$($\egraph$), where Prev$^0$($\egraph$) = $\egraph$, 
    Prev$^n$($\egraph$) = $\emptyset$ (lemma \ref{lem:finiteprev}) and 
    for each $i:0\le i \le n$ Prev$^i$($\egraph$) is consistent (lemma \ref{lem:consprev}).
    We prove the lemma by using induction on the sequence $\tt{P}$.

  \noindent  \textbf{Base Case:}  For  Prev$^n$($\egraph$) = $\emptyset$, it holds trivially.
    
   \noindent \textbf{Inductive Hypothesis:} Assume it holds for Prev$^k$($\egraph$). We show it holds for for the Prev$^{k-1}$($\egraph$).
    From lemma \ref{lem:reachblity}, it follows that Prev$^{k-1}$($\egraph$) is reachable since Prev$^k$($\egraph$) = Prev(Prev$^{k-1}$($\egraph$)) is reachable.
\end{proof}

\begin{theorem}
    The DPOR algorithm is complete.
    \label{thm:compl}
\end{theorem}

\begin{proof}
    From lemma \ref{lem:prevseq}, it follows that for any consistent execution graph $\egraph$ is explored by the DPOR algorithm by generating a sequence of reachable graphs to which $\egraph$ belongs.
 \end{proof}

\subsection{Optimality}
\label{app:optmal}
\begin{theorem}
    DPOR algorithm is optimal, that is it generates any execution graph  $\egraph$ exactly once.
    \label{thm:optml}
\end{theorem}
\begin{proof}
    Assume that the algorithm generates two sequences $s_1$ and $s_2$ such that $\egraph$ is reachable starting from $\egraph_\emptyset$.
    Let $\egraph_{p_1}$ be Prev($\egraph$) in $s_1$ and $\egraph_{p_2}$ be Prev($\egraph$) in $s_2$. 
    To prove the lemma, we need to show that $\egraph_{p_1}$ $\equiv$ $\egraph_{p_2}$.
    Let $ev = \lastofg{\egraph}$. 
    We consider two cases:
    \begin{itemize}
        \item $\reversed(ev)$. Assume $\delayedrf(\egraph,r)$ and $w=RF(e)$. 
        Since $\reversed(ev)$ and $\egraph_{p_1}$ $\egraph_{p_2}$ are consistent, 
        it follows that $\egraph_{p_1}|_{\egraph \backslash r}$ $\equiv$ $\egraph_{p_1}|_{\egraph \backslash r}$ and 
        for event $e$ $ e \in \egraph_{p_1} \backslash \egraph$ iff  $ e \in \egraph_{p_2} \backslash \egraph$.
         $\isoptimal()$ ensures that for each read $r' \in \egraph_{p_1} \cap \egraph_{p_1}$, 
        there exists a maximal write $w'$ such that $(w' 
        ~\rf_{\varof{r}}$$ ~r')$$ \in $$\egraph_{p_1}$ 
        and $(w' $$~\rf_{\varof{r}}$$ ~r')$$ \in $$\egraph_{p_2}$. 

        \item $\neg(\reversed(ev))$. 
        Since $\neg(\reversed(ev))$, we have $\egraph_{p_1} = \egraph_{p_2} = \egraph \backslash ev$.
    \end{itemize}
\end{proof}

\subsection{Illustration on the Prev() procedure.}
\label{app:previllustration}

Fig. \Cref{app:fig:previllust} gives the illustration for the Prev() computation for full and consistent execution graph Fig. \Cref{app:fig:previllust} (a). 
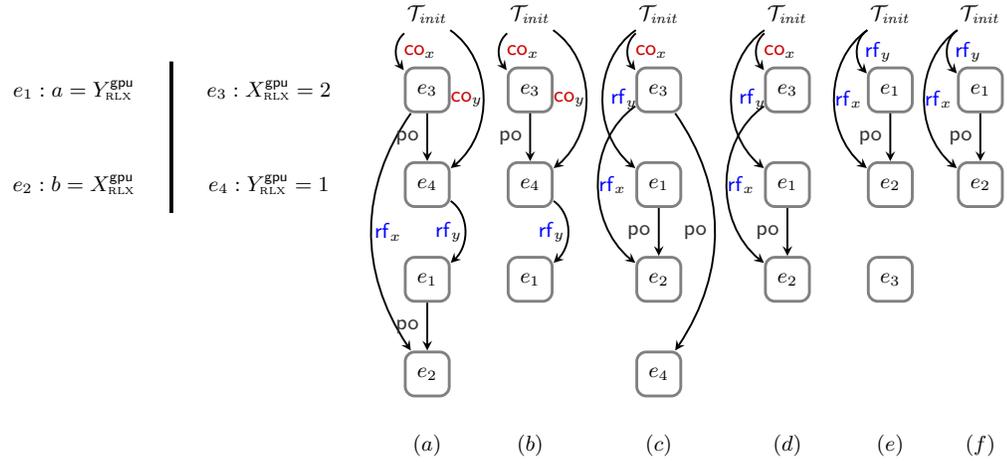
\begin{figure*}[t]
\scalebox{0.9}{%
\begin{tikzpicture}[thick, >=latex, node distance=0.2cm and 1cm,
pre/.style={<-,shorten >= 1pt, shorten <=1pt,},
post/.style={->,shorten >= 1pt, shorten <=1pt,},
und/.style={very thick, draw=gray},
node/.style={minimum size=6.5mm, fill=white!100, thick, inner sep=0},
tran/.style={box, very thick, rounded corners, draw=gray,minimum size=6.5mm, fill=white!100, inner sep=0},
virt/.style={circle,draw=black!50,fill=black!20, opacity=0}]

\newcommand{\xdisposition}{3.8}
\newcommand{\ydisposition}{3}
\newcommand{\xstep}{1.8}
\newcommand{\ystep}{0.7}
\def\bend{15}
\def\pad{0.19}

\begin{scope}[shift={(0*\xdisposition,0*\ydisposition)}]

 \draw[-, ultra thick] (0.8*\xstep, 1.6*\ystep) to  (0.8*\xstep, -1.6*\ystep);


\node	[node]		(f)	at (0*\xstep,1*\ystep) {$e_1: a=Y^{\gpu}_{\MOrlx}$};
\node	[node]		(g)	at (0*\xstep,-1*\ystep) {$e_2: b=X^{\gpu}_{\MOrlx}$};

\node	[node]		(c)	at (1.6*\xstep,1*\ystep) {$e_3: X^{\gpu}_{\MOrlx}=2$};
\node	[node]		(e)	at (1.6*\xstep,-1*\ystep) {$e_4: Y^{\gpu}_{\MOrlx}=1$};



\end{scope}

\begin{scope}[shift={(0.9*\xdisposition,0*\ydisposition)}]

\node	[node]		(i)	at (1*\xstep,2.6*\ystep) {$\inittranset$};
\node	[tran]		(f)	at (1*\xstep,-3*\ystep) {$e_1$};
\node	[tran]		(g)	at (1*\xstep,-5*\ystep) {$e_2$};

\node	[tran]		(c)	at (1*\xstep,1*\ystep) {$e_3$};
\node	[tran]		(e)	at (1*\xstep,-1*\ystep) {$e_4$};

\node	[node]		(io)	at (1*\xstep,-6.5*\ystep) {$(a)$};

\draw[->, thick, bend left=55] (e) to  node[left=-2pt, ] {$\rf_y$} (f);
\draw[->, thick, bend right=50] (i) to  node[right, ] {$\co_x$} (c);
\draw[->, thick, bend left=55] (i) to  node[left=-3pt, ] {$\co_y$} (e);
\draw[->, thick, bend right=35] (c) to  node[right=-2pt, ] {$\rf_x$} (g);
\draw[->, thick] (c) to  node[left, ] {$\po$} (e);
\draw[->, thick] (f) to  node[left, ] {$\po$} (g);

\end{scope}

\begin{scope}[shift={(1.3*\xdisposition,0*\ydisposition)}]

\node	[node]		(i)	at (1*\xstep,2.6*\ystep) {$\inittranset$};
\node	[tran]		(f)	at (1*\xstep,-3*\ystep) {$e_1$};

\node	[tran]		(c)	at (1*\xstep,1*\ystep) {$e_3$};
\node	[tran]		(e)	at (1*\xstep,-1*\ystep) {$e_4$};

\node	[node]		(io)	at (1*\xstep,-6.5*\ystep) {$(b)$};

\draw[->, thick, bend left=55] (e) to  node[left=-2pt, ] {$\rf_y$} (f);
\draw[->, thick, bend right=50] (i) to  node[right, ] {$\co_x$} (c);
\draw[->, thick, bend left=55] (i) to  node[left=-3pt, ] {$\co_y$} (e);
\draw[->, thick] (c) to  node[left, ] {$\po$} (e);

\end{scope}

\begin{scope}[shift={(1.8*\xdisposition,0*\ydisposition)}]

\node	[node]		(i)	at (1*\xstep,2.6*\ystep) {$\inittranset$};
\node	[tran]		(f)	at (1*\xstep,-1*\ystep) {$e_1$};
\node	[tran]		(g)	at (1*\xstep,-3*\ystep) {$e_2$};

\node	[tran]		(c)	at (1*\xstep,1*\ystep) {$e_3$};
\node	[tran]		(e)	at (1*\xstep,-5*\ystep) {$e_4$};

\node	[node]		(io)	at (1*\xstep,-6.5*\ystep) {$(c)$};

\draw[->, thick, bend right=55] (i) to  node[right, ] {$\rf_y$} (f);
\draw[->, thick, bend right=50] (i) to  node[right, ] {$\co_x$} (c);
\draw[->, thick, bend right=55] (c) to  node[right=-3pt, ] {$\rf_x$} (g);

\draw[->, thick, bend left=35] (c) to  node[left, ] {$\po$} (e);
\draw[->, thick] (f) to  node[left, ] {$\po$} (g);

\end{scope}

\begin{scope}[shift={(2.3*\xdisposition,0*\ydisposition)}]

\node	[node]		(i)	at (1*\xstep,2.6*\ystep) {$\inittranset$};
\node	[tran]		(f)	at (1*\xstep,-1*\ystep) {$e_1$};
\node	[tran]		(g)	at (1*\xstep,-3*\ystep) {$e_2$};

\node	[tran]		(c)	at (1*\xstep,1*\ystep) {$e_3$};

\node	[node]		(io)	at (1*\xstep,-6.5*\ystep) {$(d)$};

\draw[->, thick, bend right=55] (i) to  node[right, ] {$\rf_y$} (f);
\draw[->, thick, bend right=50] (i) to  node[right, ] {$\co_x$} (c);
\draw[->, thick, bend right=55] (c) to  node[right=-3pt, ] {$\rf_x$} (g);
\draw[->, thick] (f) to  node[left, ] {$\po$} (g);

\end{scope}

\begin{scope}[shift={(2.7*\xdisposition,0*\ydisposition)}]

\node	[node]		(i)	at (1*\xstep,2.6*\ystep) {$\inittranset$};
\node	[tran]		(f)	at (1*\xstep,1*\ystep) {$e_1$};
\node	[tran]		(g)	at (1*\xstep,-1*\ystep) {$e_2$};

\node	[tran]		(c)	at (1*\xstep,-3*\ystep) {$e_3$};

\node	[node]		(io)	at (1*\xstep,-6.5*\ystep) {$(e)$};

\draw[->, thick, bend right=55] (i) to  node[right, ] {$\rf_y$} (f);
\draw[->, thick, bend right=55] (i) to  node[right=-3pt, ] {$\rf_x$} (g);
\draw[->, thick] (f) to  node[left, ] {$\po$} (g);

\end{scope}

\begin{scope}[shift={(3.05*\xdisposition,0*\ydisposition)}]

\node	[node]		(i)	at (1*\xstep,2.6*\ystep) {$\inittranset$};
\node	[tran]		(f)	at (1*\xstep,1*\ystep) {$e_1$};
\node	[tran]		(g)	at (1*\xstep,-1*\ystep) {$e_2$};


\node	[node]		(io)	at (1*\xstep,-6.5*\ystep) {$(f)$};

\draw[->, thick, bend right=55] (i) to  node[right, ] {$\rf_y$} (f);
\draw[->, thick, bend right=55] (i) to  node[right=-3pt, ] {$\rf_x$} (g);
\draw[->, thick] (f) to  node[left, ] {$\po$} (g);

\end{scope}

\end{tikzpicture}
}
\caption{Illustration of Prev(). 
}
\label{app:fig:previllust}
\end{figure*}

\newpage
\subsection{Detailed Experiments.}
\label{app:expts}
\Cref{app:tab:lb-scalability} shows full table for the \Cref{tab:lb-scalability}.
\Cref{tab:1dconv-full,tab:GCON-full,tab:SB-full} give Time taken(in seconds) and memory consumed(in MB) to detect the heterogeneous race in 1dconv, GCON benchmarks and assertion violation in SB benchmarks (Detailed data for the \Cref{fig:scalability}).
\begin{table}[h]
\centering
 \caption{Scalability of $\projname$ on safe benchmark LB  (Time in Seconds and Memory in MB). (Events per execution)}
    \label{app:tab:lb-scalability}
\begin{tabular}{lcccc}
\hline\hline
\textbf{Program} & \textbf{Events} &  \textbf{Mem} & \textbf{Execs} &  \textbf{Time} \\
\hline
LB-2    & 26   & 84    & 3        & 0.07       \\
LB-3    & 36   & 84    & 7        & 0.03       \\
LB-4    & 46   & 84    & 15       & 0.03       \\
LB-5    & 56   & 84    & 31       & 0.03       \\
LB-6    & 66   & 84    & 63       & 0.03       \\
LB-7    & 76   & 84    & 127      & 0.04       \\
LB-8    & 86   & 84    & 255      & 0.06       \\
LB-9    & 96   & 84    & 511      & 0.11       \\
LB-10   & 106  & 84    & 1023     & 0.24       \\
LB-11   & 116  & 84    & 2047     & 0.54       \\
LB-12   & 126  & 84    & 4095     & 1.22       \\
LB-13   & 136  & 84    & 8191     & 2.82       \\
LB-14   & 146  & 84    & 16383    & 6.54       \\
LB-15   & 156  & 84    & 32767    & 15.43      \\
LB-16   & 166  & 89    & 65535    & 43.13      \\
LB-17   & 176  & 94    & 131071   & 101.24     \\
LB-18   & 186  & 101   & 262143   & 227.66     \\
LB-19   & 196  & 108   & 524287   & 511.81     \\
LB-20   & 206  & 118   & 1048575  & 1139.96  \\
LB-21   & 216  & 120   & 2097151  & 2540.37  \\
LB-22   & 226  & 127   & 4194303  & 5646.41  \\
\hline
\end{tabular}
\end{table}

\begin{table}[ht]
\begin{minipage}[t]{.3\textwidth}
    \centering
    \begin{tabular}[t]{|ccc|}
    \hline
        CTAs & Memory & Time \\
        \hline
        1  & 84.608 & 0.43  \\
        2  & 84.608 & 7.82  \\
        3  & 84.608 & 2.02  \\
        4  & 84.608 & 1.78  \\
        5  & 84.48  & 4.16  \\
        6  & 84.608 & 5.46  \\
        7  & 84.352 & 7.19  \\
        8  & 84.608 & 9.21  \\
        9  & 84.608 & 6.54  \\
        10 & 84.608 & 8.15  \\
        11 & 84.608 & 9.93  \\
        12 & 84.608 & 21.74 \\
        13 & 84.608 & 15.41 \\
        14 & 84.608 & 17.02 \\
        15 & 84.608 & 22.83 \\
        16 & 84.48  & 24.45 \\
        17 & 84.608 & 26.64 \\
        18 & 84.608 & 29.80 \\
        19 & 84.608 & 34.07 \\
        20 & 84.608 & 38.96 \\
        21 & 84.608 & 50.92 \\
        22 & 84.608 & 47.28 \\
        23 & 84.608 & 52.87 \\
        24 & 84.608 & 57.54 \\
        25 & 84.608 & 61.81 \\
        26 & 85.512 & 63.95 \\
        27 & 87.012 & 62.72 \\
        28 & 88.228 & 70.12 \\
        29 & 89.512 & 76.89 \\
        30 & 90.448 & 82.24 \\
        \bottomrule
    \end{tabular}
    \caption{1dconv benchmarks.}
    \label{tab:1dconv-full}
\end{minipage}
\hfill
\begin{minipage}[t]{0.3\textwidth}
    \centering
    \begin{tabular}[t]{|ccc|}
    \hline
        Threads & Memory & Time \\
        \hline
        1  & 84.992 & 0.24  \\
        2  & 84.864 & 0.44  \\
        3  & 84.736 & 1.01  \\
        4  & 84.864 & 0.53  \\
        5  & 84.864 & 0.59  \\
        6  & 84.992 & 0.80  \\
        7  & 84.992 & 1.02  \\
        8  & 84.864 & 1.31  \\
        9  & 84.992 & 1.65  \\
        10 & 84.864 & 2.04  \\
        11 & 84.864 & 2.48  \\
        12 & 84.992 & 3.00  \\
        13 & 84.864 & 3.62  \\
        14 & 84.992 & 4.22  \\
        15 & 84.992 & 4.92  \\
        16 & 84.992 & 5.74  \\
        17 & 84.992 & 6.60  \\
        18 & 84.992 & 7.63  \\
        19 & 84.992 & 8.66  \\
        20 & 84.992 & 9.87  \\
        21 & 84.992 & 11.50 \\
        22 & 84.992 & 13.04 \\
        23 & 84.992 & 14.69 \\
        24 & 84.992 & 16.47 \\
        25 & 84.864 & 18.26 \\
        26 & 84.992 & 20.34 \\
        27 & 84.992 & 22.43 \\
        28 & 84.992 & 24.73 \\
        29 & 84.992 & 27.13 \\
        30 & 84.992 & 29.87 \\
        31 & 84.992 & 32.36 \\
        32 & 84.992 & 35.48 \\
        33 & 84.992 & 38.41 \\
        34 & 84.992 & 41.91 \\
        35 & 84.992 & 45.37 \\
        36 & 84.992 & 49.24 \\
        37 & 84.992 & 53.77 \\
        38 & 84.992 & 58.00 \\
        39 & 84.992 & 62.42 \\
        40 & 84.992 & 66.70 \\
        41 & 84.992 & 71.93 \\
        42 & 84.992 & 76.77 \\
        43 & 84.992 & 82.59 \\
        44 & 84.992 & 87.75 \\
        45 & 84.992 & 93.86 \\
        46 & 85.436 & 101.57 \\
        47 & 86.196 & 108.03 \\
        48 & 86.892 & 114.05 \\
        49 & 87.628 & 121.86 \\
        50 & 88.388 & 128.67 \\
        \bottomrule
    \end{tabular}
    \caption{GCON benchmarks.}
    \label{tab:GCON-full}
\end{minipage}
\hfill
\begin{minipage}[t]{0.3\textwidth}
    \centering
    \begin{tabular}[t]{|ccc|}
    \hline
        Threads & Memory & Time \\
        \hline
        2  & 83.584 & 0.02 \\
        3  & 83.584 & 0.02 \\
        4  & 83.584 & 0.02 \\
        5  & 83.584 & 0.03 \\
        6  & 83.584 & 0.03 \\
        7  & 83.712 & 0.03 \\
        8  & 83.712 & 0.03 \\
        9  & 83.712 & 0.03 \\
        10 & 83.84  & 0.03 \\
        11 & 83.84  & 0.03 \\
        12 & 83.84  & 0.03 \\
        13 & 83.84  & 0.03 \\
        14 & 83.84  & 0.04 \\
        15 & 83.968 & 0.04 \\
        16 & 83.84  & 0.05 \\
        17 & 83.84  & 0.05 \\
        18 & 83.968 & 0.05 \\
        19 & 83.968 & 0.06 \\
        20 & 84.096 & 0.06 \\
        21 & 84.096 & 0.07 \\
        22 & 86.016 & 0.08 \\
        23 & 93.416 & 0.08 \\
        24 & 95.852 & 0.10 \\
        25 & 98.944 & 3.02 \\
        \bottomrule
    \end{tabular}
    \caption{SB benchmarks}
    \label{tab:SB-full}
\end{minipage}
\label{tab:scalability full}
\end{table}

\end{document}